\documentclass[12pt]{article}
\usepackage{verbatim,amsmath, amsfonts, amssymb, amsthm, natbib, float, dsfont, bm}
\usepackage{graphicx, algorithmic, algorithm, enumerate,color,hyperref, multirow}
\usepackage{pifont}
\usepackage{xr}
\def\real{\mathbb R}
\def\hc{\mathcal{\hat C}}

\newcommand\mydots{\ifmmode\ldots\else\makebox[1em][c]{.\hfil.\hfil.}\fi}
\newcommand{\indep}{\rotatebox[origin=c]{90}{$\models$}}

\newtheorem{prop}{Proposition}
\newtheorem{theorem}{Theorem}
\newtheorem{lemma}{Lemma}

\newtheorem{definition}{Definition}
\newtheorem{corollary}{Corollary}
\newtheorem{assumption}{Assumption}

\newcommand{\xmark}{\ding{55}}%

\makeatletter
\newcommand*{\addFileDependency}[1]{% argument=file name and extension
  \typeout{(#1)}
  \@addtofilelist{#1}
  \IfFileExists{#1}{}{\typeout{No file #1.}}
}
\makeatother
 
\newcommand*{\myexternaldocument}[1]{%
    \externaldocument{#1}%
    \addFileDependency{#1.tex}%
    \addFileDependency{#1.aux}%
}

\myexternaldocument{supplement}

\begin{document}

\def\spacingset#1{\renewcommand{\baselinestretch}%
{#1}\small\normalsize} \spacingset{1}
\newcommand*{\Scale}[2][4]{\scalebox{#1}{$#2$}}%

\title{Selective Inference for Hierarchical Clustering}
\author{Lucy L. Gao$^{\dagger}$\footnote{Corresponding author: lucy.gao@stat.ubc.ca}, Jacob Bien$^{\circ}$, and Daniela Witten $^{\ddagger}$ \\~\\
{\small $\dagger$ Department of Statistics, University of British Columbia} \\ 
{\small $\circ$ Department of Data Sciences and Operations, University of Southern California} \\ 
{\small $\ddagger$ Departments of Statistics and Biostatistics, University of Washington}  \\
}

\maketitle
\begin{abstract}
Classical tests for a difference in means control the type I error rate when the groups are defined \emph{a priori}. However, when the groups are instead defined via clustering, then applying a classical test yields an extremely inflated type I error rate. Notably, this problem persists even if two separate and independent data sets are used to define the groups and to test for a difference in their means. To address this problem, in this paper, we propose a selective inference approach to test for a difference in means between two clusters. Our procedure controls the selective type I error rate by accounting for the fact that the choice of null hypothesis was made based on the data. We describe how to efficiently compute exact p-values for clusters obtained using agglomerative hierarchical clustering with many commonly-used linkages. We apply our method to simulated data and to single-cell RNA-sequencing data. 
\end{abstract}

\vfill

\newpage
\spacingset{1.5} % DON'T change the spacing!

\section{Introduction} 
\label{sec:intro}

Testing for a difference in means between groups is fundamental to answering research questions across virtually every scientific area. Classical tests control the type I error rate when the groups are defined \emph{a priori}. However, it is increasingly common for researchers to instead define the groups via a clustering algorithm. In the context of single-cell RNA-sequencing data, researchers often cluster the cells to identify putative cell types, then test for a difference in means between the putative cell types in that same data set; see e.g. \citet{hwang2018single}. Unfortunately, available tests do not properly account for the double-use of data, which invalidates the resulting inference.  One example of this problematic issue can be seen in the \verb+FindMarkers+ function in the popular and highly-cited \verb+R+ package \verb+Seurat+ \citep{seuratv1, seuratv2, seuratv3, seuratv4}. Many recent papers have described the issues associated with the use of data for both clustering and downstream testing, without proposing suitable solutions \citep{luecken2019current, lahnemann2020eleven, deconinck2021recent}. In fact, testing for a difference in means between a pair of estimated clusters while controlling the type I error rate has even been described as one of eleven ``grand challenges" in the entire field of single-cell data science \citep{lahnemann2020eleven}. Similar issues arise in the field of neuroscience \citep{kriegeskorte2009circular}. 

In this paper, we develop a valid test for a difference in means between two clusters estimated from the data. We consider the following model for $n$ observations of $q$ features:
\begin{align} 
{\bf X} \sim \mathcal{MN}_{n \times q}(\bm \mu, {\bf I}_n, \sigma^2 {\bf I}_q), \label{eq:mod}
\end{align} 
where $\bm \mu \in \real^{n \times q}$, with rows $\mu_i$, is unknown, and $\sigma^2 > 0$ is known. (We discuss the case where $\sigma^2$ is unknown in Section \ref{sec:unknown}.) For $\mathcal{G}  \subseteq \{1,2, \ldots, n\}$, let
\begin{align} 
\bar{\mu}_{\mathcal{G}} \equiv \frac{1}{|\mathcal{G}|} \sum \limits_{i \in \mathcal{G}} \mu_i \quad \text{and} \quad    \bar{X}_{\mathcal{G}} \equiv \frac{1}{|\mathcal{G}|} \sum \limits_{i \in \mathcal{G}} X_i, \label{eq:def-mu}
\end{align} 
which we refer to as the mean of $\mathcal{G}$ and the empirical mean of $\mathcal{G}$ in ${\bf X}$, respectively. Given a realization ${\bf x} \in \mathbb{R}^{n \times q}$ of ${\bf X}$, we first apply a clustering algorithm $\mathcal{C}$ to obtain $\mathcal{C}({\bf x})$, a partition of $\{1, 2, \ldots, n\}$. We then use ${\bf x}$ to test, for a pair of clusters $\hc_1, \hc_2 \in \mathcal{C}({\bf x})$, 
\begin{align} H_0^{\{\hc_1, \hc_2\}}: \bar{\mu}_{\hc_1} = \bar{\mu}_{\hc_2} \quad \text{versus} \quad H_1^{\{\hc_1, \hc_2\}}: \bar{\mu}_{\hc_1} \neq  \bar{\mu}_{\hc_2}. \label{eq:null} 
\end{align}
It is tempting to simply apply a Wald test of
\eqref{eq:null}, with p-value given by
\begin{align} \mathbb{P}_{H_0^{\{\hc_1, \hc_2\}}} \left ( \|\bar{X}_{\hc_1} - \bar{X}_{\hc_2} \|_2 \geq \|\bar{x}_{\hc_1} - \bar{x}_{\hc_2}\|_2 \right ), \label{eq:wald}
\end{align}
where $ \|\bar{X}_{\hc_1} - \bar{X}_{\hc_2} \|_2 \overset{H_0^{\{\hc_1, \hc_2\}}}{\sim} \left ( \sigma \sqrt{\frac{1}{|\hc_1|} + \frac{1}{|\hc_2|}} \right )  \cdot \chi_q$. However, since we clustered ${\bf x}$ to get $\mathcal{C}({\bf x}) = \{\hc_k\}_{k=1}^K$, we will observe substantial differences between $\{\bar{x}_{\hc_k}\}_{k=1}^K$ even when there is no signal in the data, as is shown in Figure \ref{fig:pval-naive}(a). That is, in \eqref{eq:wald}, the random variable on the left-hand side of the inequality follows a scaled $\chi_q$ distribution, but the right-hand side of the inequality is \emph{not} drawn from a scaled $\chi_q$ distribution, because $\hc_1$ and $\hc_2$ are functions of ${\bf x}$. In short, the problem is that we used the data to select a null hypothesis to test. Since the Wald test does not account for this hypothesis selection procedure, it is extremely anti-conservative, as is shown in Figure \ref{fig:pval-naive}(b).

\begin{figure}[h!]
\centering
\includegraphics[scale=0.5]{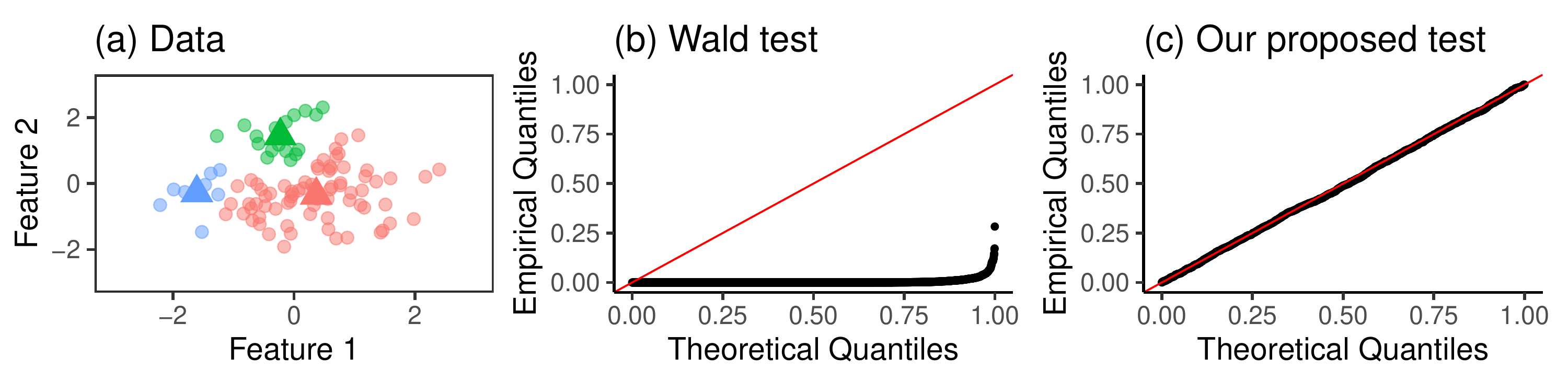}
\caption{\label{fig:pval-naive} (a) A simulated data set from \eqref{eq:mod} with $\bm \mu = \bm 0_{100 \times 2}$ and $\sigma^2 = 1$. We apply average linkage hierarchical clustering to get three clusters. The empirical means (defined in \eqref{eq:def-mu}) of the three clusters are displayed as triangles. QQ-plots of the Uniform$(0, 1)$ distribution against the p-values from (b) the Wald test in \eqref{eq:wald} and (c) our proposed test, over 2000 simulated data sets from \eqref{eq:mod} with $\bm \mu = \bm 0_{100 \times 2}$ and $\sigma^2 = 1$. For each simulated data set, a p-value was computed for a randomly chosen pair of estimated clusters. }
\end{figure} 

At first glance, it seems that we might be able to overcome this problem via sample splitting. That is, we  divide the observations into a training and a test set, cluster the observations in the training set,  and then assign the test set observations to those clusters, as in Figures \ref{fig:split}(a)--(c). Then, we  apply the Wald test in \eqref{eq:wald} to the test set. Unfortunately, by assigning test observations to clusters, we have once again used the data to select a null hypothesis to test, in the sense that $H_0^{\{\hc_1, \hc_2\}}$ in \eqref{eq:null} is a function of the test observations. Thus, the Wald test is extremely anti-conservative (Figure \ref{fig:split}(d)). In other words, sample splitting does not provide a valid way to test the hypothesis in \eqref{eq:null}. 

\begin{figure}[h!]
\centering
\includegraphics[scale=0.5]{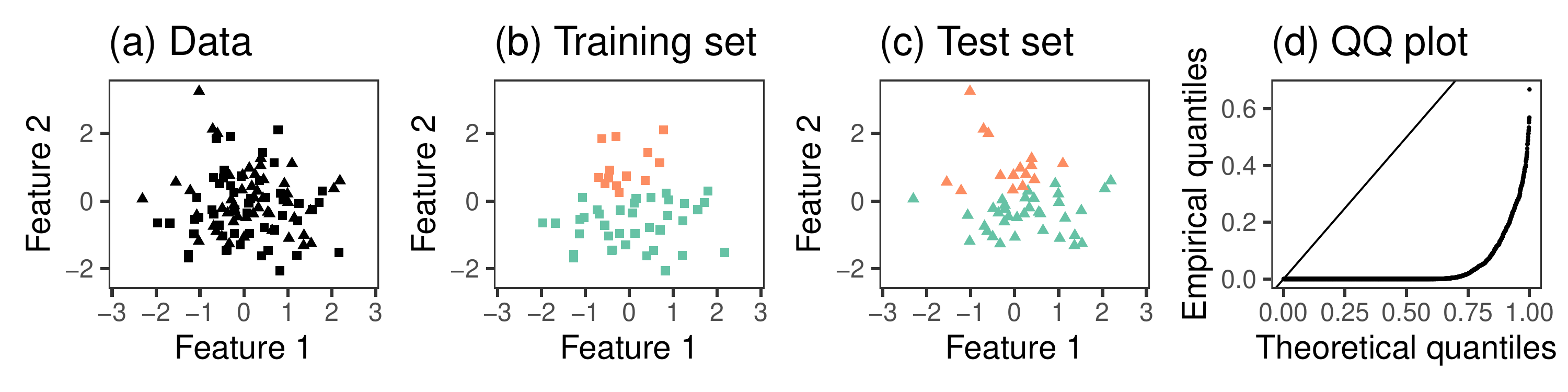}
\caption{\label{fig:split} (a) A simulated data set from \eqref{eq:mod} with $\bm \mu = \bm 0_{100 \times 2}$ and $\sigma^2 = 1$.  (b) We cluster the training set using average linkage hierarchical clustering. (c) We assign clusters in the test set by training a 3-nearest neighbors classifier on the training set. (d) QQ-plot of the Uniform$(0, 1)$ distribution against the Wald p-values in the test set, over 2000 simulated data sets for which each cluster in the test set was assigned at least one observation. }
\end{figure} 

In this paper, we develop a \emph{selective inference} framework to test for a difference in means after clustering. This framework exploits ideas from the recent literature on selective inference for regression and changepoint detection \citep{fithian2014optimal, loftus2015selective,  lee2016exact,yang2016selective, hyun2018exact,  jewell2019testing, mehrizi2021valid}. The key idea is as follows: since we chose to test $H_0^{\{\hc_1, \hc_2\}}: \bar{\mu}_{\hc_1} = \bar{\mu}_{\hc_2}$  because $\hc_1, \hc_2 \in \mathcal{C}({\bf x})$, we can account for this hypothesis selection procedure by defining a p-value that conditions on the event $\{\hc_1, \hc_2 \in \mathcal{C}({\bf X})\}$. This yields a correctly-sized test, as seen in Figure \ref{fig:pval-naive}(c). 

A large body of work  evaluates the statistical significance of a clustering by testing the goodness-of-fit of models under the misspecification of the number of clusters \citep{chen2001modified, liu2008statistical, chen2012inference, maitra2012bootstrapping, kimes2017statistical} or by assessing the stability of estimated clusters \citep{suzuki2006pvclust}. Most of these papers conduct bootstrap sampling or asymptotic approximations to the null distribution. Our proposed framework avoids the need for resampling and provides exact finite-sample inference for the difference in means between  a pair of estimated clusters, under the assumption that $\sigma$ in \eqref{eq:mod} is known. Chapter 3 of \citet{campbell2018statistical} considers testing for a difference in means after convex clustering \citep{hocking2011clusterpath}, a relatively esoteric form of clustering. Our framework is particularly efficient when applied to hierarchical clustering, which is one of the most popular types of clustering across a number of fields. \citet{zhang2019valid} proposes splitting the data, clustering the training set, and applying these clusters to the test set as illustrated in Figures \ref{fig:split}(a)--(c). They develop a selective inference framework that yields valid p-values for a difference in the mean of a single feature between two clusters in the test set. Our framework avoids the need for sample splitting, and thereby allows inference on the set of clusters obtained from \emph{all} (rather than a subset of) the data.

The rest of the paper is organized as follows. In Section \ref{sec:frame}, we develop a framework to test for a difference in means after clustering. We apply this framework to compute p-values for hierarchical clustering in Section \ref{sec:hier}.  We describe extensions, simulation results, and applications to real data in Sections \ref{sec:practical}, \ref{sec:sim}, and \ref{sec:app}. The discussion is in Section \ref{sec:discuss}. 

\section{Selective inference for clustering}
\label{sec:frame}

\subsection{A test of no difference in means between two clusters}
\label{sec:frame-test}
Let ${\bf x} \in \mathbb{R}^{n \times q}$ be an arbitrary realization from \eqref{eq:mod}, and let $\hc_1$ and $\hc_2$ be an arbitrary pair of clusters in $\mathcal{C}({\bf x})$. 
Since we chose to test $H_0^{\{\hc_1, \hc_2\}}$ because $\hc_1, \hc_2 \in \mathcal{C}({\bf x})$, it is natural to define the p-value as a conditional version of \eqref{eq:wald},
\begin{align} 
\mathbb{P}_{H_0^{\{\hc_1, \hc_2\}}}\left ( \|\bar{X}_{\hc_1} - \bar{X}_{\hc_2}\|_2 \geq \|\bar{x}_{\hc_1} - \bar{x}_{\hc_2} \|_2 ~\Big |  ~\hc_1, \hc_2 \in  \mathcal{C}({\bf X})  \right ). \label{eq:pval-hard} 
\end{align} 
This amounts to asking, ``Among all realizations of ${\bf X}$ that result in clusters $\hc_1$ and $\hc_2$, what proportion have a difference in cluster means at least as large as the difference in cluster means in our observed data set, when in truth $\bar{\mu}_{\hc_1} = \bar{\mu}_{\hc_2}$?" One can show that rejecting $H_0^{\{\hc_1, \hc_2\}}$ when \eqref{eq:pval-hard} is below $\alpha$ controls the \emph{selective type I error rate} \citep{fithian2014optimal} at level $\alpha$.
\begin{definition}[Selective type I error rate for clustering] 
\label{def:selective}
For any non-overlapping groups of observations  $\mathcal{G}_1, \mathcal{G}_2 \subseteq \{1, 2, \ldots, n\}$, let $H_0^{\{\mathcal{G}_1, \mathcal{G}_2\}}$ denote the null hypothesis that $ \bar{\mu}_{\mathcal{G}_1} = \bar{\mu}_{\mathcal{G}_2}$. We say that a test of $H_0^{\{\mathcal{G}_1, \mathcal{G}_2\}}$ based on ${\bf X}$ controls the selective type I error rate for clustering at level $\alpha$ if
\begin{align} \mathbb{P}_{H_0^{\{\mathcal{G}_1, \mathcal{G}_2\}}} \left (\text{reject } H_0^{\{\mathcal{G}_1, \mathcal{G}_2\}} \text{ based on } {\bf X} \text{ at level } \alpha ~\Big | ~ \mathcal{G}_1, \mathcal{G}_2  \in \mathcal{C}({\bf X}) \right ) \leq \alpha, \quad \forall~ 0 \leq \alpha \leq 1. \label{eq:selective} 
\end{align} 
That is, if $H_0^{\{\mathcal{G}_1, \mathcal{G}_2\}}$ is true, then the conditional probability of rejecting $H_0^{\{\mathcal{G}_1, \mathcal{G}_2\}}$ based on ${\bf X}$ at level $\alpha$, given that $\mathcal{G}_1$ and $\mathcal{G}_2$ are clusters in $\mathcal{C}({\bf X})$, is bounded by $\alpha$. 
\end{definition}
However, \eqref{eq:pval-hard} cannot be calculated, since the conditional distribution of $\|\bar{X}_{\hc_1} - \bar{X}_{\hc_2} \|_2$ given $\hc_1, \hc_2 \in  \mathcal{C}({\bf X})$ involves the unknown nuisance parameters $\bm \pi_{\nu(\hc_1, \hc_2)}^\perp \bm \mu$, where $\bm \pi_{\nu}^\perp = {\bf I}_n - \frac{\nu \nu^T}{\| \nu\|_2^2}$ projects onto the orthogonal complement of the vector $\nu$, and where 
\begin{align} 
[\nu(\hc_1, \hc_2)]_i = \mathds{1} \{i \in \hc_1\}/|\hc_1| - \mathds{1} \{i \in \hc_2\}/|\hc_2| \label{eq:def-nu}. 
\end{align} 
 In other words, it requires knowing aspects of $\bm \mu$ that are not known under the null. Instead, we will define the p-value for $H_0^{\{\hc_1, \hc_2\}}$ in \eqref{eq:null} to be 
\begin{align}
\Scale[0.9]{p({\bf x}; \{\hc_1, \hc_2\}  )} &= \Scale[0.9]{\mathbb{P}_{H_0^{\{\hc_1, \hc_2\}}} \Big ( \|\bar{X}_{\hc_1} - \bar{X}_{\hc_2}\|_2 \geq  \|\bar{x}_{\hc_1} - \bar{x}_{\hc_2}\|_2 ~ \Big | ~\hc_1, \hc_2 \in \mathcal{C}({\bf X}), \bm \pi_{\nu(\hc_1, \hc_2)}^\perp {\bf X} = \bm \pi_{\nu(\hc_1, \hc_2)}^\perp {\bf x}}, \nonumber \\ 
&\hspace{72mm} \Scale[0.9]{\text{dir} \left (\bar{X}_{\hc_1} - \bar{X}_{\hc_2} \right ) = \text{dir}\left (\bar{x}_{\hc_1} - \bar{x}_{\hc_2} \right )  \Big )}, \label{eq:pval-def}
\end{align}
 where $\text{dir}(w) = \frac{w}{\|w\|_2} \mathds{1} \{w \neq 0\}$.
The following result shows that conditioning on these additional events makes \eqref{eq:pval-def} computationally tractable by constraining the randomness in ${\bf X}$ to a scalar random variable, while maintaining control of the selective type I error rate. 
\begin{theorem} 
\label{thm:pval}
For any realization ${\bf x}$ from \eqref{eq:mod} and for any non-overlapping groups of observations $\mathcal{G}_1, \mathcal{G}_2 \subseteq \{1,2, \ldots, n\}$, 
\begin{align} 
p({\bf x}; \{\mathcal{G}_1, \mathcal{G}_2\}) = 1 - \mathbb{F} \left (\|\bar{x}_{\mathcal{G}_1} - \bar{x}_{\mathcal{G}_2}\|_2  ; \sigma \sqrt{\frac{1}{|\mathcal{G}_1|} + \frac{1}{|\mathcal{G}_2|}}, \mathcal{S}({\bf x}; \{ \mathcal{G}_1, \mathcal{G}_2 \} ) \right ), \label{eq:pval-phi}
\end{align} 
where $p(\cdot; \cdot)$ is defined in \eqref{eq:pval-def}, $\mathbb{F}(t; c, \mathcal{S})$ denotes the cumulative distribution function of a $c \cdot \chi_q$ random variable truncated to the set $\mathcal{S}$, and
\begin{align} 
\Scale[0.9]{\mathcal{S}({\bf x}; \{\mathcal{G}_1, \mathcal{G}_2\} ) = \left \{ \phi \geq 0: \mathcal{G}_1, \mathcal{G}_2 \in  \mathcal{C} \left (\bm \pi_{\nu(\mathcal{G}_1, \mathcal{G}_2)}^\perp {\bf x} + \left ( \frac{\phi}{\frac{1}{|\mathcal{G}_1|} + \frac{1}{|\mathcal{G}_2|}} \right ) \nu(\mathcal{G}_1, \mathcal{G}_2) \text{dir}(\bar{x}_{\mathcal{G}_1} - \bar{x}_{\mathcal{G}_2})^T  \right ) \right \}}.  \label{eq:defS-general}
\end{align} 
Furthermore, if $H_0^{\{\mathcal{G}_1, \mathcal{G}_2\}}$ is true, then 
\begin{align} 
\mathbb{P}_{H_0^{\{\mathcal{G}_1, \mathcal{G}_2\}}}\left ( p({\bf X}; \{\mathcal{G}_1, \mathcal{G}_2\}) \leq \alpha ~\Big |  ~\mathcal{G}_1, \mathcal{G}_2 \in  \mathcal{C}({\bf X})  \right ) = \alpha, \quad \forall ~ 0 \leq \alpha \leq 1. \label{eq:selective-ours}
\end{align} 
That is, rejecting $H_0^{\{\mathcal{G}_1, \mathcal{G}_2\}}$ whenever $p({\bf x}; \{\mathcal{G}_1, \mathcal{G}_2\})$ is below $\alpha$ controls the selective type I error rate (Definition \ref{def:selective}) at level $\alpha$.
\end{theorem} 
We prove Theorem \ref{thm:pval} in Appendix \ref{sec:proof-pval}. 
 It follows from \eqref{eq:pval-phi} that to compute the p-value $p({\bf x}; \{\hc_1, \hc_2\})$ in \eqref{eq:pval-def}, it suffices to characterize the one-dimensional set 
\begin{align} 
\mathcal{\hat S} \equiv \mathcal{S}({\bf x}; \{\hc_1, \hc_2\}) = \{ \phi \geq 0 : \hc_1, \hc_2  \in  \mathcal{C}({\bf x}'(\phi)) \}, \label{eq:defS}
\end{align} 
where $\mathcal{S}({\bf x}; \cdot)$ is defined in \eqref{eq:defS-general}, and where 
\begin{align} {\bf x}'(\phi) = \bm \pi_{\hat \nu}^\perp {\bf x} + \left ( \frac{\phi}{1/|\hc_1| + 1/|\hc_2|} \right ) \hat \nu ~ \text{dir}(\bar{x}_{\hc_1} - \bar{x}_{\hc_2})^T, \quad \hat \nu = \nu(\hc_1, \hc_2), \label{eq:xphi} 
\end{align} 
for $\nu(\cdot, \cdot)$ defined in \eqref{eq:def-nu}. 

While the test based on \eqref{eq:pval-def} controls the selective type I error rate, the extra conditioning may lead to lower power than a test based on \eqref{eq:pval-hard}\citep{lee2016exact, jewell2019testing, mehrizi2021valid}. However, \eqref{eq:pval-def} has a major advantage over \eqref{eq:pval-hard}: Theorem \ref{thm:pval} reveals that computing \eqref{eq:pval-def} simply requires characterizing $\mathcal{\hat S}$ in \eqref{eq:defS}. This is the focus of Section \ref{sec:hier}. 

\subsection{Interpreting ${\bf x}'(\phi)$ and $\mathcal{\hat S}$}
\label{sec:phi-picture}
Since ${\bf x}^T \hat \nu = \bar{x}_{\hc_1} - \bar{x}_{\hc_2}$, where $\bar{x}_{\hc_1}$ is defined in \eqref{eq:def-mu} and $\hat \nu$ is defined in \eqref{eq:xphi}, it follows that the $i$th row of ${\bf x}'(\phi)$ in \eqref{eq:xphi} is 
\begin{align} 
[{\bf x}'(\phi)]_i = \begin{cases} x_i + \left ( \frac{|\hc_2|}{|\hc_1| + |\hc_2|} \right ) \left ( \phi - \|\bar{x}_{\hc_1} - \bar{x}_{\hc_2}\|_2 \right )\text{dir}\left (\bar{x}_{\hc_1} - \bar{x}_{\hc_2} \right ) , & \text{if } i \in \hc_1, \\ 
x_i  - \left ( \frac{|\hc_1|}{|\hc_1| + |\hc_2|} \right ) \left ( \phi - \|\bar{x}_{\hc_1} - \bar{x}_{\hc_2}\|_2 \right )\text{dir}\left (\bar{x}_{\hc_1} - \bar{x}_{\hc_2} \right ), & \text{if } i \in \hc_2, \\ 
x_i, & \text{if } i \not \in \hc_1 \cup \hc_2.
\end{cases} \label{eq:xphirow}
\end{align}
We can interpret ${\bf x}'(\phi)$ as a perturbed version of ${\bf x}$, where observations in clusters $\hc_1$ and $\hc_2$ have been ``pulled apart" (if $\phi > \|\bar{x}_{\hc_1} - \bar{x}_{\hc_2}\|_2$) or ``pushed together" (if $0 \leq \phi < \|\bar{x}_{\hc_1} - \bar{x}_{\hc_2}\|_2$)  in the direction of $ \bar{x}_{\hc_1} - \bar{x}_{\hc_2}$. Furthermore, $\mathcal{\hat S}$ in \eqref{eq:defS} describes the set of non-negative $\phi$ for which applying the clustering algorithm $\mathcal{C}$ to the perturbed data set ${\bf x}'(\phi)$ yields $\hc_1$ and $\hc_2$. To illustrate this interpretation, we apply average linkage hierarchical clustering to a realization from  \eqref{eq:mod} to obtain three clusters. Figure \ref{fig:xphi}(a)-(c) displays ${\bf x} = {\bf x}'(\phi)$ for  $\phi = \|\bar{x}_{\hc_1} - \bar{x}_{\hc_{2}}\|_2 = 4$, along with  ${\bf x}'(\phi)$ for $\phi = 0$ and $\phi = 8$, respectively. The clusters $\hc_1, \hc_{2} \in \mathcal{C}({\bf x})$ are shown in blue and orange. In Figure \ref{fig:xphi}(b), since $\phi=0$, the blue and orange clusters have been ``pushed together" so that there is no difference between their empirical means, and average linkage hierarchical clustering no longer estimates these clusters.  By contrast, in Figure \ref{fig:xphi}(c), the blue and orange clusters have been ``pulled apart", and average linkage hierarchical clustering still estimates these clusters. In this example, $\mathcal{\hat S}  = [2.8, \infty)$. 
 
\begin{figure}
\centering
\includegraphics[scale=0.4]{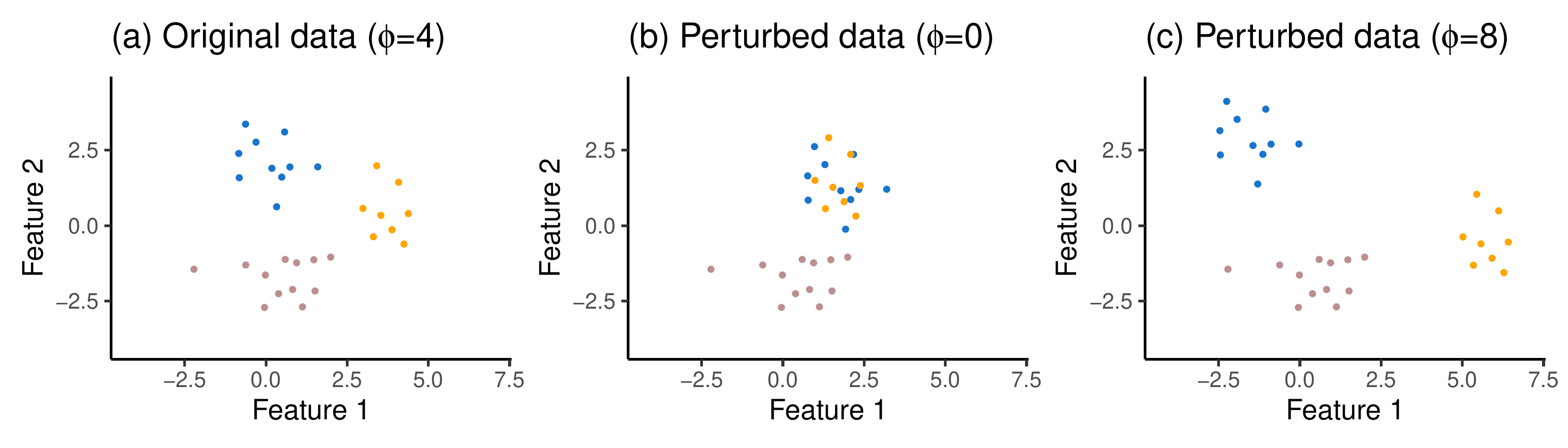}
\caption{\label{fig:xphi} The observations belonging to $\hc_1, \hc_{2}\in \mathcal{C}({\bf x})$ are displayed in blue and orange for: (a) the original data set ${\bf x} = {\bf x}'(\phi)$ with $\phi = \|\bar{x}_{\hc_1} - \bar{x}_{\hc_{2}}\|_2 = 4$, (b) a perturbed data set  ${\bf x}'(\phi)$ with $\phi = 0$, and (c) a perturbed data set  ${\bf x}'(\phi)$ with $\phi = 8$.  }
\end{figure} 

\section{Computing  $\mathcal{\hat S}$ for hierarchical clustering} 
\label{sec:hier}
We now  consider computing $\mathcal{\hat S}$ defined in \eqref{eq:defS} for clusters defined via hierarchical clustering. After reviewing hierarchical clustering (Section \ref{sec:review}), we explicitly characterize $\mathcal{\hat S}$ (Section \ref{sec:key}), and then show how specific properties, such as the dissimilarity and linkage used, lead to substantial computational savings in computing $\mathcal{\hat S}$ (Sections \ref{sec:lance}--\ref{sec:single}). 
\subsection{A brief review of agglomerative hierarchical clustering} 
\label{sec:review}

Agglomerative hierarchical clustering produces a sequence of clusterings. The first clustering, $\mathcal{C}^{(1)}({\bf x})$,   contains $n$ clusters, each with a single observation. The $(t+1)$th clustering, $\mathcal{C}^{(t+1)}({\bf x})$,   is created by merging the two most similar (or least dissimilar) clusters in the $t$th clustering, for $t = 1, \ldots, n-1$. Details are provided in Algorithm~\ref{alg:hier}. 

\begin{algorithm}[h!]  
\caption{Agglomerative hierarchical clustering of a data set ${\bf x}$ \label{alg:hier}} 
~ Let $\mathcal{C}^{(1)}({\bf x}) = \{ \{1\}, \{2\}, \ldots, \{n\} \}$. For $t = 1, \ldots, n-1$:
\begin{enumerate} 
\item Define $\left \{\mathcal{W}_1^{(t)}({\bf x}), \mathcal{W}_2^{(t)}({\bf x}) \right \} = \underset{\mathcal{G}, \mathcal{G}' \in \mathcal{C}^{(t)}({\bf x}),  \mathcal{G} \neq \mathcal{G}'}{\arg\min}  d(\mathcal{G}, \mathcal{G}'; {\bf x})$. (We assume throughout this paper that the minimizer is unique.) 
\item Merge $\mathcal{W}_1^{(t)}({\bf x})$ and $\mathcal{W}_2^{(t)}({\bf x})$ at the height of $d \left (\mathcal{W}_1^{(t)}({\bf x}), \mathcal{W}_2^{(t)}({\bf x}); {\bf x} \right )$ in the dendrogram, and let $\mathcal{C}^{(t+1)}({\bf x}) =  \mathcal{C}^{(t)}({\bf x}) \cup \left \{\mathcal{W}_1^{(t)}({\bf x}) \cup \mathcal{W}_2^{(t)}({\bf x})\right \} \backslash \left \{\mathcal{W}_1^{(t)}({\bf x}), \mathcal{W}_2^{(t)}({\bf x}) \right \}.$
\end{enumerate} 
\end{algorithm}

Algorithm \ref{alg:hier} involves a function $d(\mathcal{G}, \mathcal{G}'; {\bf x})$, which quantifies the dissimilarity between two groups of observations. We assume throughout this paper that the dissimilarity between the $i$th and $i$'th observations, $d(\{i\}, \{i'\}; {\bf x})$, depends on the data through $x_i - x_i'$ only. For example, we could define $d(\{i\}, \{i'\}; {\bf x}) = \|x_i - x_{i'}\|_2^2$. When $\max \{ |\mathcal{G}|, |\mathcal{G}'|\} > 1$, then we extend the pairwise similarity to the dissimilarity between groups of obervations using the notion of \emph{linakge}, to be discussed further in Section \ref{sec:lance}.

\subsection{An explicit characterization of  $\mathcal{\hat S}$  for hierarchical clustering}
\label{sec:key} 
We saw in Sections~\ref{sec:frame-test}--\ref{sec:phi-picture} that to compute the p-value $p({\bf x}; \{\hc_1, \hc_2\})$ defined in \eqref{eq:pval-def}, we must characterize the set  $\mathcal{\hat S} = \{ \phi \geq 0 : \hc_1, \hc_2 \in \mathcal{C}({\bf x}'(\phi)) \}$ in \eqref{eq:defS}, where ${\bf x}'(\phi)$ in \eqref{eq:xphi} is a perturbed version of  ${\bf x}$ in which observations in the clusters $\hc_1$ and $\hc_2$ have been ``pulled together" or ``pushed apart". We do so now for hierarchical clustering. 
\begin{lemma} 
\label{lem:heights}
Suppose that $\mathcal{C} = \mathcal{C}^{(n-K+1)}$, i.e. we perform hierarchical clustering to obtain $K$ clusters. Then, 
\begin{align} 
d\left ( \mathcal{W}_1^{(t)}({\bf x}), \mathcal{W}_2^{(t)}({\bf x}); {\bf x}'(\phi) \right ) = d\left ( \mathcal{W}_1^{(t)}({\bf x}), \mathcal{W}_2^{(t)}({\bf x}); {\bf x} \right ),  ~\forall ~ \phi \geq 0, \forall~ t = 1,  \ldots, n-K, \label{eq:heights}
\end{align} 
where $\left (\mathcal{W}_1^{(t)}({\bf x}),\mathcal{W}_2^{(t)}({\bf x}) \right )$ is the ``winning pair" of clusters that merged at the $t^{th}$ step of the hierarchical clustering algorithm applied to ${\bf x}$. Furthermore, for any $\phi \geq 0$,
\begin{gather} 
\hc_1, \hc_2 \in \mathcal{C}({\bf x}'(\phi)) \quad \text{ if and only if } \quad \mathcal{C}^{(t)}({\bf x}'(\phi)) = \mathcal{C}^{(t)}({\bf x}) ~  \forall ~ t = 1,  \ldots, n-K+1. \label{eq:merges} 
\end{gather} 
\end{lemma}
We prove Lemma \ref{lem:heights} in Appendix \ref{sec:proof-hier}.
The right-hand side of \eqref{eq:merges} says that the same merges occur in the first $n-K$ steps of the hierarchical clustering algorithm applied to ${\bf x}'(\phi)$ and ${\bf x}$. To characterize the set of merges that occur in ${\bf x}$, consider the set of all ``losing pairs", i.e. all cluster pairs that co-exist but are not the ``winning pair" in the first $n-K$ steps: 
\begin{align} 
\mathcal{L}({\bf x}) = \bigcup \limits_{t=1}^{n-K} \left \{ \{\mathcal{G}, \mathcal{G}'\}: \mathcal{G}, \mathcal{G}' \in \mathcal{C}^{(t)}({\bf x}), \mathcal{G} \neq \mathcal{G}', \{ \mathcal{G}, \mathcal{G}' \} \neq \left \{ \mathcal{W}_1^{(t)}({\bf x}), \mathcal{W}_2^{(t)}({\bf x}) \right \} \right \}. \label{eq:losers}
\end{align} 
Each pair $\{\mathcal{G}, \mathcal{G}' \} \in \mathcal{L}({\bf x})$ has a ``lifetime" that starts at the step where both have been created, $l_{\mathcal{G}, \mathcal{G}'}({\bf x}) \equiv \min \left \{1 \leq t \leq n-K: \mathcal{G}, \mathcal{G}' \in \mathcal{C}^{(t)}({\bf x}), \{\mathcal{G}, \mathcal{G}' \} \neq \{ \mathcal{W}_1^{(t)}({\bf x}), \mathcal{W}_2^{(t)}({\bf x})\}  \right \}$, and ends at step
$u_{\mathcal{G}, \mathcal{G}'}({\bf x}) \equiv  \max \left \{ 1 \leq t \leq n-K: \mathcal{G}, \mathcal{G}' \in \mathcal{C}^{(t)}({\bf x}), \{\mathcal{G}, \mathcal{G}' \} \neq \{ \mathcal{W}_1^{(t)}({\bf x}), \mathcal{W}_2^{(t)}({\bf x})\}  \right \}.$ By construction, each pair $\{\mathcal{G}, \mathcal{G}' \} \in \mathcal{L}({\bf x})$ is never the winning pair at any point in its lifetime, i.e. $d(\mathcal{G}, \mathcal{G}'; {\bf x}) > d \left (\mathcal{W}_1^{(t)}({\bf x}), \mathcal{W}_2^{(t)}({\bf x}); {\bf x}\right )$ for all $l_{\mathcal{G},\mathcal{G}'}({\bf x}) \leq t \leq u_{\mathcal{G},\mathcal{G}'}({\bf x})$. Therefore, ${\bf x}'(\phi)$ preserves the merges that occur in the first $n-K$ steps in ${\bf x}$ if and only if $d(\mathcal{G}, \mathcal{G}'; {\bf x}'(\phi)) > d \left (\mathcal{W}_1^{(t)}({\bf x}), \mathcal{W}_2^{(t)}({\bf x}); {\bf x}'(\phi) \right )$ for all $l_{\mathcal{G},\mathcal{G}'}({\bf x}) \leq t \leq u_{\mathcal{G},\mathcal{G}'}({\bf x})$ and for all $\{\mathcal{G}, \mathcal{G}' \} \in \mathcal{L}({\bf x})$. Furthermore, \eqref{eq:heights} says that $d \left (\mathcal{W}_1^{(t)}({\bf x}), \mathcal{W}_2^{(t)}({\bf x}); {\bf x}'(\phi) \right ) = d \left (\mathcal{W}_1^{(t)}({\bf x}), \mathcal{W}_2^{(t)}({\bf x}); {\bf x} \right )$ for all $\phi \geq 0$ and $1 \leq t \leq n-K$.  This leads to the following result. 
\begin{theorem} 
\label{thm:hierS}
Suppose that $\mathcal{C} = \mathcal{C}^{(n-K+1)}$, i.e. we perform hierarchical clustering  to obtain $K$ clusters. Then, for $\mathcal{\hat S}$ defined in \eqref{eq:defS},
\begin{align} 
\mathcal{\hat S} = \bigcap \limits_{\{\mathcal{G}, \mathcal{G}'\} \in \mathcal{L}({\bf x})} \left \{ \phi \geq 0: d(\mathcal{G}, \mathcal{G}'; {\bf x}'(\phi)) > \underset{l_{\mathcal{G},\mathcal{G}'}({\bf x}) \leq t \leq u_{\mathcal{G},\mathcal{G}'}({\bf x})}{\max} d\left ( \mathcal{W}_1^{(t)}({\bf x}),\mathcal{W}_2^{(t)}({\bf x})  ; {\bf x} \right ) \right \}, \label{eq:hierS}
\end{align} 
where $\left \{ \mathcal{W}_1^{(t)}({\bf x}),\mathcal{W}_2^{(t)}({\bf x}) \right \}$ is the pair of clusters that merged at the $t^{th}$ step of the hierarchical clustering algorithm applied to ${\bf x}$, $\mathcal{L}({\bf x})$ is defined in \eqref{eq:losers} to be the set of ``losing pairs" of clusters in ${\bf x}$, and $[l_{\mathcal{G},\mathcal{G}'}({\bf x}) , u_{\mathcal{G},\mathcal{G}'}({\bf x}) ]$ is the lifetime of such a pair of clusters in ${\bf x}$. Furthermore, \eqref{eq:hierS} is the intersection of $\mathcal{O}(n^2)$ sets. 
\end{theorem} 
We prove Theorem \ref{thm:hierS} in Appendix \ref{sec:proofS}.
 Theorem \ref{thm:hierS} expresses $\mathcal{\hat S}$  in \eqref{eq:defS} as the intersection of $\mathcal{O}(n^2)$ sets of the form $ \{\phi \geq 0: d(\mathcal{G}, \mathcal{G}'; {\bf x}'(\phi)) > h_{\mathcal{G}, \mathcal{G}'}({\bf x}) \}$, where 
\begin{align} 
h_{\mathcal{G}, \mathcal{G}'} ({\bf x}) \equiv \underset{l_{\mathcal{G},\mathcal{G}'}({\bf x}) \leq t \leq u_{\mathcal{G},\mathcal{G}'}({\bf x})}{\max} d\left ( \mathcal{W}_1^{(t)}({\bf x}),\mathcal{W}_2^{(t)}({\bf x}); {\bf x} \right ) \label{eq:hab}
\end{align} 
is the maximum merge height in the dendrogram of ${\bf x}$ over the lifetime of $\{\mathcal{G}, \mathcal{G}'\}$. The next subsection is devoted to understanding when and how these sets can be efficiently computed. In particular, by specializing to squared Euclidean distance and a certain class of linkages, we will show that each of these sets is defined by a single quadratic inequality, and that the coefficients of all of these quadratic inequalities can be efficiently computed.

\subsection{Squared Euclidean distance and 	``linear update" linkages} 
\label{sec:lance}
Consider hierarchical clustering with squared Euclidean distance and a linkage that satisfies a linear Lance-Williams update \citep{lance1967general} of the form 
\begin{align} 
d(\mathcal{G}_1 \cup \mathcal{G}_2, \mathcal{G}_3; {\bf x})&= \alpha_1 d(\mathcal{G}_1, \mathcal{G}_3; {\bf x}) + \alpha_2 d(\mathcal{G}_2, \mathcal{G}_3; {\bf x}) + \beta d(\mathcal{G}_1, \mathcal{G}_2; {\bf x}). \label{eq:lance1}
\end{align}
This includes average, weighted, Ward, centroid, and median linkage (Table \ref{tab:lance}). 
\begin{table}[h!]
\begin{tabular}{l|lllllll}
                            & Average    & Weighted & Ward       & Centroid   & Median     & Single     & Complete   \\ \hline
Satisfies \eqref{eq:lance1} & \checkmark & \checkmark       & \checkmark & \checkmark & \checkmark & \xmark     & \xmark    \\ 
\hline
\begin{tabular}{@{}l@{}}Does not produce \\inversions \end{tabular}
   & \checkmark & \checkmark   & \checkmark  & \xmark   & \xmark      & \checkmark  & \checkmark  
\end{tabular}
\caption{\label{tab:lance} Properties of seven linkages in the case of squared Euclidean distance \citep{murtagh2012algorithms}. Table 1 of \citet{murtagh2012algorithms} specifies $\alpha_1, \alpha_2$, and $\beta$ in \eqref{eq:lance1}. }
\end{table}

We have seen in Section \ref{sec:key} that to evaluate \eqref{eq:hierS}, we must evaluate $\mathcal{O}(n^2)$ sets of the form $ \{\phi \geq 0: d(\mathcal{G}, \mathcal{G}'; {\bf x}'(\phi)) > h_{\mathcal{G}, \mathcal{G}'}({\bf x}) \}$ with $\{\mathcal{G}, \mathcal{G}' \} \in \mathcal{L}({\bf x})$, where $\mathcal{L}({\bf x})$ in \eqref{eq:losers} is the set of losing cluster pairs in ${\bf x}$. We now present results needed to characterize these sets. 
\begin{lemma} 
\label{prop:quad}
Suppose that we define $d(\{i\}, \{i'\}; {\bf x}) = \|x_i - x_{i'}\|_2^2$. Then, for all $i \neq i'$, $d(\{i\}, \{i'\}; {\bf x}'(\phi)) = a_{ii'} \phi^2 + b_{ii'}\phi  + c_{ii'}$, 
where for $\hat \nu$ defined in \eqref{eq:xphi}, $a_{ii'} = \left ( \frac{\hat \nu_i - \hat \nu_{i'}}{\|\hat \nu\|_2^2} \right )^2$,  \\ 
$b_{ii'} = 2 \left ( \left (  \frac{\hat \nu_i - \hat \nu_{i'}}{\|\hat \nu\|_2^2} \right )  \langle \text{dir}({\bf x}^T \hat \nu), x_i - x_{i'} \rangle - a_{ii'} \|{\bf x}^T \hat \nu\|_2\right ) $, and $c_{ii'} = \left \|x_i - x_{i'} - \left (  \frac{\hat \nu_i - \hat \nu_{i'}}{\|\hat \nu\|_2^2} \right ) ({\bf x}^T \hat \nu) \right \|_2^2$. 
\end{lemma} 
Lemma \ref{prop:quad} follows directly from the definition of ${\bf x}'(\phi)$ in \eqref{eq:xphi}, and 
does not require \eqref{eq:lance1} to hold. Next, we specialize to squared Euclidean distance \emph{and} linkages satisfying \eqref{eq:lance1}, and characterize $d(\mathcal{G}, \mathcal{G}'; {\bf x}'(\phi))$, the dissimilarity between pairs of clusters in ${\bf x}'(\phi)$. The following result follows immediately from Lemma \ref{prop:quad} and the fact that linear combinations of quadratic functions of $\phi$ are also quadratic functions of $\phi$. 
\begin{prop} 
\label{prop:average}
Suppose we define $d(\{i\}, \{i'\}; {\bf x}) = \|x_i - x_{i'}\|_2^2$, and we define $d(\mathcal{G}, \mathcal{G}'; {\bf x})$ using a linkage that satisfies \eqref{eq:lance1}. Then, 
 $d(\mathcal{G}, \mathcal{G}'; {\bf x}'(\phi))$ is a quadratic function of $\phi$ for all $\mathcal{G} \neq \mathcal{G}'$. Furthermore, given the coefficients corresponding to the quadratic functions $d(\mathcal{G}_1, \mathcal{G}_3; {\bf x}'(\phi))$,  $d(\mathcal{G}_2, \mathcal{G}_3; {\bf x}'(\phi))$, and $d(\mathcal{G}_1, \mathcal{G}_2; {\bf x}'(\phi))$, we can compute the coefficients corresponding to the quadratic function $d(\mathcal{G}_1 \cup \mathcal{G}_2, \mathcal{G}_3; {\bf x}'(\phi))$ in $\mathcal{O}(1)$ time, using \eqref{eq:lance1}. 
\end{prop}
Lastly, we characterize the cost of computing $h_{\mathcal{G}, \mathcal{G}'}({\bf x})$ in \eqref{eq:hab}. Naively, computing $h_{\mathcal{G}, \mathcal{G}'}({\bf x})$ could require $\mathcal{O}(n)$ operations. However, if the dendrogram of ${\bf x}$ has no inversions below the $(n-K)$th merge, i.e. if 
$d\left ( \mathcal{W}_1^{(t)}({\bf x}),\mathcal{W}_2^{(t)}({\bf x}); {\bf x} \right ) < d \left ( \mathcal{W}_1^{(t+1)}({\bf x}),\mathcal{W}_2^{(t+1)}({\bf x}); {\bf x} \right ) $ for all $t < n-K$, then
$ h_{\mathcal{G}, \mathcal{G}'}({\bf x}) = d \left ( \mathcal{W}_1^{ \left (u_{\mathcal{G}, \mathcal{G}'}({\bf x}) \right )}({\bf x}),\mathcal{W}_2^{ \left (u_{\mathcal{G}, \mathcal{G}'}({\bf x}) \right )}({\bf x})  ; {\bf x} \right )$. More generally, $h_{\mathcal{G}, \mathcal{G}'}({\bf x}) = \underset{t \in \mathcal{M}_{\mathcal{G}, \mathcal{G}'}({\bf x}) \cup \{u_{\mathcal{G}, \mathcal{G}'}({\bf x}) \}}{\max} d \left (\mathcal{W}_1^{(t)}({\bf x}),   \mathcal{W}_2^{(t)}({\bf x}); {\bf x} \right ), $
where  
$\mathcal{M}_{\mathcal{G}, \mathcal{G}'}({\bf x}) = \Big \{t : l_{\mathcal{G}, \mathcal{G}'}({\bf x}) \leq t < u_{\mathcal{G}, \mathcal{G}'}({\bf x}),  d\left ( \mathcal{W}_1^{(t)}({\bf x}),\mathcal{W}_2^{(t)}({\bf x}); {\bf x} \right ) > d \left ( \mathcal{W}_1^{(t+1)}({\bf x}),\mathcal{W}_2^{(t+1)}({\bf x}); {\bf x} \right )  \Big \}$  is the set of steps where inversions occur in the dendrogram of ${\bf x}$ during the lifetime of the cluster pair $\{\mathcal{G}, \mathcal{G}'\}$.  This leads to the following result. 
\begin{prop} 
\label{lem:computeH}
For any $\{\mathcal{G}, \mathcal{G}'\} \in \mathcal{L}({\bf x})$, given its lifetime $l_{\mathcal{G}, \mathcal{G}'}({\bf x})$ and $u_{\mathcal{G}, \mathcal{G}'}({\bf x})$, and given $\mathcal{M}({\bf x}) = \left \{1 \leq t \leq n-K : d\left ( \mathcal{W}_1^{(t)}({\bf x}),\mathcal{W}_2^{(t)}({\bf x}); {\bf x} \right ) < d \left ( \mathcal{W}_1^{(t+1)}({\bf x}),\mathcal{W}_2^{(t+1)}({\bf x}); {\bf x} \right )  \right \}$, i.e. the set of steps where inversions occur in the dendrogram of ${\bf x}$ below the $(n-K)$th merge, we can compute $h_{\mathcal{G}, \mathcal{G}'}({\bf x})$  in $\mathcal{O}(|\mathcal{M}({\bf x})|+1)$ time. 
\end{prop} 
We prove Proposition \ref{lem:computeH} in Appendix \ref{sec:proof-computeH}.
 Proposition \ref{lem:computeH}  does not require defining $d(\{i\}, \{i'\}; {\bf x}) = \|x_i - x_{i'}\|_2^2$ and does not require \eqref{eq:lance1} to hold. We now characterize the cost of computing $\mathcal{\hat S}$ defined in \eqref{eq:defS}, in the case of  
squared Euclidean distance and linkages that satisfy \eqref{eq:lance1}. 
\begin{prop} 
\label{prop:lance} 
Suppose we define $d(\{i\}, \{i'\}; {\bf x}) = \|x_i - x_{i'}\|_2^2$, and we define $d(\mathcal{G}, \mathcal{G}'; {\bf x})$ using a linkage that satisfies \eqref{eq:lance1}. Then, we can compute $\mathcal{\hat S}$ defined in \eqref{eq:defS} in $\mathcal{O}\Big ((|\mathcal{M}({\bf x})| + \log(n)) n^2 \Big )$ time. 
\end{prop} 
A detailed algorithm for computing $\mathcal{\hat S}$ is provided in Appendix \ref{sec:lance-algo}. If the linkage we use does not produce inversions, then $|\mathcal{M}({\bf x})| = 0$ for all ${\bf x}$. Average, weighted, and Ward linkage satisfy \eqref{eq:lance1} and are guaranteed not to produce inversions (Table \ref{tab:lance}), thus ${\mathcal{\hat S}}$ can be computed in $\mathcal{O}(n^2 \log(n))$ time.

\subsection{Squared Euclidean distance and single linkage} 
\label{sec:single}
Single linkage does not satisfy \eqref{eq:lance1}  (Table \ref{tab:lance}), and the inequality that defines the set $ \{\phi \geq 0: d(\mathcal{G}, \mathcal{G}'; {\bf x}'(\phi)) > h_{\mathcal{G}, \mathcal{G}'}({\bf x}) \}$ is not quadratic in $\phi$ for $|\mathcal{G}| > 1$ or $|\mathcal{G}'| > 1$. Consequently, in the case of single linkage with squared Euclidean distance, we cannot efficiently evaluate the expression of $\mathcal{\hat S}$  in \eqref{eq:hierS} using the approach outlined in Section \ref{sec:lance}. 

Fortunately, the definition of single linkage leads to an even simpler expression of $\mathcal{\hat S}$ than \eqref{eq:hierS}. Single linkage defines $d(\mathcal{G}, \mathcal{G}'; {\bf x}) = \underset{i \in \mathcal{G}, i' \in \mathcal{G}'}{\min}~ d(\{i\}, \{i'\}; {\bf x})$. Applying this definition to \eqref{eq:hierS} yields 
$\mathcal{\hat S} = \bigcap \limits_{\{\mathcal{G}, \mathcal{G}'\} \in \mathcal{L}({\bf x})} \bigcap \limits_{i \in \mathcal{G}} \bigcap \limits_{i' \in \mathcal{G}'} \left \{ \phi \geq 0 : d(\{i\}, \{i'\}; {\bf x}'(\phi)) >h_{\mathcal{G}, \mathcal{G}'}({\bf x}) \right \},$ where $\mathcal{L}({\bf x})$ in \eqref{eq:losers} is the set of losing cluster pairs in ${\bf x}$. Therefore, in the case of single linkage,
$ \mathcal{\hat S} = \bigcap \limits_{\{i, i'\} \in \mathcal{L}'({\bf x})} \mathcal{\hat S}_{i, i'},$  
where $\mathcal{L}'({\bf x}) = \{ \{i, i'\}: i \in \mathcal{G}, i' \in \mathcal{G}', \{\mathcal{G}, \mathcal{G}' \} \in \mathcal{L}({\bf x})\}$ and \\
$\mathcal{\hat S}_{i, i'} = \Big \{ \phi \geq 0: d(\{i\}, \{i'\}; {\bf x}'(\phi)) > \underset{\{\mathcal{G}, \mathcal{G}'\} \in \mathcal{L}({\bf x}): i \in \mathcal{G}, i' \in \mathcal{G}'}{\max}~ h_{\mathcal{G}, \mathcal{G}'}({\bf x}) \Big \}.$ The following result characterizes the sets of the form $\mathcal{\hat S}_{i, i'}$.
\begin{prop} 
\label{prop:single}
Suppose that $\mathcal{C} = \mathcal{C}^{(n-K+1)}$ , i.e. we perform hierarchical clustering to obtain $K$ clusters. Let $i \neq i'$. If $i, i' \in \hc_1$ or $i, i' \in \hc_2$ or $i, i' \not  \in \hc_1 \cup \hc_2$, then $\mathcal{\hat S}_{i, i'} = [0, \infty)$. Otherwise, 
$\mathcal{\hat S}_{i, i'} = \left \{ \phi \geq 0: d(\{i\}, \{i'\}; {\bf x}'(\phi)) > d\left (\mathcal{W}^{(n-K)}_1({\bf x}), \mathcal{W}^{(n-K)}_2({\bf x}); {\bf x} \right ) \right \}. $
\end{prop} 
We prove Proposition \ref{prop:single} in Appendix \ref{sec:proof-single}.
Therefore, 
\begin{align} 
\mathcal{\hat S} &= \bigcap \limits_{\{i, i'\} \in \mathcal{L}'({\bf x})} \mathcal{\hat S}_{i, i'} \nonumber \\ 
&=   \bigcap \limits_{\{i, i'\} \in \mathcal{I}({\bf x})} \left \{ \phi \geq 0: ~ d(\{i\}, \{i'\}; {\bf x}'(\phi)) >  d\left (\mathcal{W}^{(n-K)}_1({\bf x}), \mathcal{W}^{(n-K)}_2({\bf x}); {\bf x} \right ) \right \}, \label{eq:Ssingle}
\end{align} 
where $\mathcal{I}({\bf x}) = \mathcal{L}'({\bf x}) \backslash \left [ \left \{\{i, i'\}: i, i' \in \hc_1 \right \} \cup \left \{\{i, i'\}: i, i' \in \hc_2 \right \} \cup \left \{\{i, i'\}: i, i' \not \in \hc_1 \cup \hc_2 \right \} \right ] $. Recall from Lemma \ref{prop:quad} that in the case of squared Euclidean distance, $d(\{i\}, \{i'\}; {\bf x}'(\phi))$ is a quadratic function of $\phi$ whose coefficients can be computed in $\mathcal{O}(1)$ time. Furthermore,  $\mathcal{O}(n^2)$ sets are intersected in \eqref{eq:Ssingle}. Therefore, we can evaluate \eqref{eq:Ssingle} in $\mathcal{O}(n^2\log (n))$ time by solving $\mathcal{O}(n^2)$ quadratic inequalities and intersecting their solution sets. 

\section{Extensions} 
\label{sec:practical}
\subsection{Monte Carlo approximation to  the p-value }
\label{sec:approx}
We may be interested in computing the p-value $p({\bf x}; \{\hc_1, \hc_2\})$ defined in \eqref{eq:pval-def} for clustering methods 
 where $\mathcal{\hat S}$ in \eqref{eq:defS} cannot be efficiently computed (e.g. complete linkage hierarchical clustering or non-hierarchical clustering methods). Thus, we develop a Monte Carlo approximation to the p-value that does not require us to compute $\mathcal{\hat S}$. Recalling from \eqref{eq:defS} that $\mathcal{\hat S} = \mathcal{S}({\bf x}; \{\hc_1, \hc_2\})$, it follows from Theorem \ref{thm:pval} that 
\begin{align} 
\Scale[0.95]{ p({\bf x};  \{\hc_1, \hc_2\} ) = \mathbb{E} \left [ \mathds{1} \left \{ \phi \geq  \|\bar{x}_{\hc_1} - \bar{x}_{\hc_2}\|_2, \hc_1, \hc_2 \in  \mathcal{C} ({\bf x}'(\phi)) \right \} \right ] / \mathbb{E} \left [  \mathds{1} \left \{ \hc_1, \hc_2 \in  \mathcal{C} ({\bf x}'(\phi)) \right \}\right ],} \label{eq:rat}
\end{align}
for $\phi \sim \sigma \left (\sqrt{\frac{1}{|\hc_1|} + \frac{1}{|\hc_2|}} \right )\cdot \chi_q$, and for ${\bf x}'(\phi)$ defined in \eqref{eq:xphi}. Thus, we could naively sample $\phi_1, \ldots, \phi_N \overset{i.i.d.}{\sim} \sigma \left (\sqrt{\frac{1}{|\hc_1|} + \frac{1}{|\hc_2|}} \right )\cdot \chi_q$, and approximate the expectations in \eqref{eq:rat} with averages over the samples.
However, when $\|\bar{x}_{\hc_1} - \bar{x}_{\hc_2}\|_2$ is in the tail of the $\sigma \left (\sqrt{\frac{1}{|\hc_1|} + \frac{1}{|\hc_2|}} \right )\cdot \chi_q$ distribution, the naive approximation of the expectations in \eqref{eq:rat} is poor for finite values of $N$. Instead, we use an importance sampling approach, as in \citet{yang2016selective}. We sample $\omega_1, \ldots, \omega_N \overset{i.i.d}{\sim} N \left (\|\bar{x}_{\hc_1} - \bar{x}_{\hc_2} \|_2, \sigma^2\left (\frac{1}{|\hc_1|} + \frac{1}{|\hc_2|} \right ) \right )$, and approximate \eqref{eq:rat} with 
\begin{align*} 
\Scale[0.95]{p({\bf x};  \{\hc_1, \hc_2\} ) \approx  \left ( \sum \limits_{i=1}^N \pi_i \mathds{1}\left \{ \omega_i \geq  \|\bar{x}_{\hc_1} - \bar{x}_{\hc_2}\|_2, \hc_1, \hc_2 \in  \mathcal{C} ({\bf x}'(\omega_i)) \right \} \right ) \big / \left (\sum \limits_{i=1}^N \pi_i \mathds{1} \left \{\hc_1, \hc_2 \in  \mathcal{C} ({\bf x}'(\omega_i)) \right \} \right )},
\end{align*} 
for $\pi_i = \frac{f_1(\omega_i)}{f_2(\omega_i)}$, where $f_1$ is the density of a $\sigma \left (\sqrt{\frac{1}{|\hc_1|} + \frac{1}{|\hc_2|}} \right )\cdot \chi_q$ random variable, and $f_2$ is the density of a $N \left (\|\bar{x}_{\hc_1} - \bar{x}_{\hc_2} \|_2, \sigma^2\left (\frac{1}{|\hc_1|} + \frac{1}{|\hc_2|} \right ) \right )$ random variable. 

\subsection{Non-spherical covariance matrix}
\label{sec:general-cov} 
The selective inference framework in Section \ref{sec:frame} assumes that ${\bf x}$ is a realization from \eqref{eq:mod}, so that $\text{Cov}(X_i) = \sigma^2 {\bf I}_q$. In this subsection, we instead assume that ${\bf x}$ is a realization from 
\begin{align} 
{\bf X} \sim \mathcal{MN}_{n \times q}(\bm \mu, {\bf I}_n,  \bm \Sigma) \label{eq:mod2}, 
\end{align} 
where $\bm \Sigma$ is a known $q \times q$ positive definite matrix. We define the p-value of interest as
\begin{align*} 
&\Scale[0.9]{p_\Sigma({\bf x}; \{\hc_1, \hc_2\})  = \mathbb{P}_{H_0^{\{\hc_1, \hc_2\}}} \Bigg ( \|\bm \Sigma^{-\frac{1}{2}}(\bar{X}_{\hc_1} - \bar{X}_{\hc_2})\|_2^2 \geq \|\bm \Sigma^{-\frac{1}{2}}(\bar{x}_{\hc_1} - \bar{x}_{\hc_2})\|_2^2~\Big |  ~ \hc_1, \hc_2 \in  \mathcal{C}({\bf X}), }  \\ 
&\hspace{46mm}\Scale[0.9]{\bm \pi_{\nu(\hc_1, \hc_2)}^\perp {\bf X}  =  \bm \pi_{\nu(\hc_1, \hc_2)}^\perp {\bf x}, ~ \text{dir}\left ( \bm \Sigma^{-1/2} (\bar{X}_{\hc_1} - \bar{X}_{\hc_2}) \right ) = \text{dir}\left ( \bm \Sigma^{-1/2} (\bar{x}_{\hc_1} - \bar{x}_{\hc_2}) \right )  \Bigg ). }
\end{align*} 
\begin{theorem} 
\label{thm:pval-general}
For any realization ${\bf x}$ from \eqref{eq:mod2}, and for any non-overlapping groups of observations $\mathcal{G}_1, \mathcal{G}_2 \subseteq \{1, 2, \ldots, n\}$, 
\begin{align} 
p_\Sigma({\bf x}; \{\mathcal{G}_1, \mathcal{G}_2\}) = 1 - \mathbb{F}\left (\|\bm \Sigma^{\scriptscriptstyle -\frac{1}{2}} {\bf x}^T  \nu(\mathcal{G}_1, \mathcal{G}_2)\|_2; \|\nu(\mathcal{G}_1, \mathcal{G}_2)\|_2, \mathcal{S}_\Sigma({\bf x}; \{\mathcal{G}_1, \mathcal{G}_2\}) \right ), \label{eq:pval-phi-general}
\end{align} 
where $\mathcal{S}_\Sigma({\bf x}; \{\mathcal{G}_1, \mathcal{G}_2\}) =  \left \{ \phi \geq 0:  \mathcal{G}_1, \mathcal{G}_2  \in \mathcal{C}\left ( \bm \pi^\perp_{\nu(\mathcal{G}_1, \mathcal{G}_2)} {\bf x} + \phi \left (\frac{\nu(\mathcal{G}_1, \mathcal{G}_2)}{\|\nu(\mathcal{G}_1, \mathcal{G}_2)\|_2^2} \right ) \text{dir}(  \bm \Sigma^{\scriptscriptstyle -\frac{1}{2}} {\bf x}^T  \nu(\mathcal{G}_1, \mathcal{G}_2))^T \bm \Sigma^{\scriptscriptstyle \frac{1}{2}} \right ) \right \}$.
 Furthermore, if $H_0^{\{\mathcal{G}_1, \mathcal{G}_2\}}$ is true, then 
 $$ \mathbb{P}_{H_0^{\{\mathcal{G}_1, \mathcal{G}_2\}}}(p_\Sigma({\bf X}; \{\mathcal{G}_1, \mathcal{G}_2\}) \leq \alpha \mid \mathcal{G}_1, \mathcal{G}_2 \in \mathcal{C}({\bf X})) = \alpha, \quad \text{for all } 0 \leq \alpha \leq 1.$$ 
 That is, rejecting $H_0^{\{\mathcal{G}_1, \mathcal{G}_2\}}$ whenever $p_\Sigma({\bf X}; \{\mathcal{G}_1, \mathcal{G}_2\})$ is below $\alpha$ controls the selective type I error rate (Definition \ref{def:selective}). 
\end{theorem} 
We omit the proof of Theorem \ref{thm:pval-general}, since it closely follows that of Theorem \ref{thm:pval}. In the case of hierachical clustering with squared Euclidean distance, we can adapt Sections \ref{sec:lance}--\ref{sec:single} to compute $\mathcal{S}_\Sigma({\bf x}; \{\hc_1, \hc_2\})$ by replacing $a_{ii'}$, $b_{ii'}$, and $c_{ii'}$ in Lemma \ref{prop:quad} with
$\tilde a_{ii'} = \left ( \frac{\hat \nu_i - \hat \nu_{i'}}{\|\hat \nu\|_2^2}\right )^2 \left ( \frac{\|{\bf x}^T \hat \nu\|_2}{\|\bm \Sigma^{-1/2} {\bf x}^T \hat \nu \|_2}\right )^2,$ 
$\tilde b_{ii'}= 2 \Big ( \Big ( \frac{\hat \nu_i - \hat \nu_{i'}}{\|\hat \nu\|_2^2} \Big ) \Big ( \frac{\|{\bf x}^T \hat \nu\|_2}{\|\bm \Sigma^{-1/2} {\bf x}^T \hat \nu \|_2}\Big  ) \langle \text{dir}({\bf x}^T \hat \nu), x_i - x_{i'} \rangle - \tilde a_{ii'} \|\bm \Sigma^{-1/2}{\bf x}^T \hat \nu\|_2\Big )$, and  $\tilde c_{ii'} = \left \|x_i - x_{i'} - \left ( \frac{\hat \nu_i - \hat \nu_{i'}}{\|\hat \nu\|_2^2}\right )({\bf x}^T \hat \nu) \right \|_2^2$.
 
 \subsection{Unknown variance} 
\label{sec:unknown} 
The selective inference framework in Section \ref{sec:frame} assumes that $\sigma$ in model \eqref{eq:mod} is known. If $\sigma$ is unknown, then we can plug an estimate of $\sigma$ into \eqref{eq:pval-phi}, as follows:
\begin{align} 
\hat p \left ({\bf x}; \{\hc_1, \hc_2\} \right ) = 1 - \mathbb{F} \left (\|\bar{x}_{\hc_1} - \bar{x}_{\hc_2}\|_2; \hat \sigma({\bf x}) \sqrt{\frac{1}{|\hc_1|} + \frac{1}{|\hc_2|}}, ~ \mathcal{\hat S} \right ).\label{eq:est-pval}
\end{align} 
Intuitively, if we plug in a consistent estimate of $\sigma$, we might expect to asymptotically control the selective type I error rate. For example, this is true in the context of selective inference for low-dimensional linear regression \citep{tian2017asymptotics}, where the error variance can be consistently estimated under mild assumptions. 

Unfortunately, consistently estimating $\sigma$ in \eqref{eq:mod} presents a major challenge. Similar issues arise when estimating the error variance in high-dimensional linear regression; see e.g.  \citet{reid2016study}. Thus, we adopt a similar approach to \citet{tibshirani2018uniform}, which studied the theoretical implications of plugging in a simple over-estimate of the error variance in the context of selective inference for high-dimensional linear regression. In the following result, we show that plugging in an asymptotic over-estimate of $\sigma$ in \eqref{eq:est-pval} leads to asymptotic control of the selective type I error rate (Definition \ref{def:selective}). 
\begin{theorem} 
\label{thm:plug-in-conservative} 
For $n = 1, 2, \ldots$, suppose that ${\bf X}^{(n)} \sim \mathcal{MN}_{n \times q}(\bm \mu^{(n)}, {\bf I}_n, \sigma^2 {\bf I}_q)$. Let ${\bf x}^{(n)}$ be a realization of ${\bf X}^{(n)}$, and $\hc_1^{(n)}$ and $\hc_2^{(n)}$ be a pair of clusters estimated from ${\bf x}^{(n)}$. Suppose that 
$\underset{n \rightarrow \infty}{\lim} ~ \mathbb{P}_{H_0^{ \left \{\hc_1^{(n)}, \hc_2^{(n)} \right \}}}\left ( \hat \sigma({\bf X}^{(n)}) \geq \sigma ~\Big |~ \hc_1^{(n)}, \hc_2^{(n)} \in \mathcal{C}({\bf X}^{(n)}) \right ) = 1$. Then, for any $\alpha \in [0, 1]$, we have that $ \underset{n \rightarrow \infty}{\lim} ~\mathbb{P}_{H_0^{ \left \{\hc_1^{(n)}, \hc_2^{(n)} \right \}}} \left (\hat p({\bf X}^{(n)}; \{\hc_1^{(n)}, \hc_2^{(n)}\}) \leq \alpha ~\Big |~ \hc_1^{(n)}, \hc_2^{(n)} \in \mathcal{C}({\bf X}^{(n)}) \right )  \leq  \alpha$. 
\end{theorem} 
We prove Theorem \ref{thm:plug-in-conservative} in Appendix \ref{sec:proof-plug-in}. In Appendix \ref{sec:supp-unknown}, we provide an estimator of $\sigma$ that satisfies the conditions in Theorem \ref{thm:plug-in-conservative}.

\subsection{Consequences of selective type I error rate control} 
\label{sec:consequence}
This paper focuses on developing tests of $H_0^{\{\mathcal{G}_1, \mathcal{G}_2\}}: \bar{\mu}_{\mathcal{G}_1} = \bar{\mu}_{\mathcal{G}_2}$ that control the selective type I error rate (Definition \ref{def:selective}). However, it is cumbersome to demonstrate selective type I error rate control via simulation, as $\mathbb{P}(\mathcal{G}_1, \mathcal{G}_2 \in \mathcal{C}({\bf X}))$ can be small when $H_0^{\{\mathcal{G}_1, \mathcal{G}_2\}}$ holds. 

Nevertheless, two related properties can be demonstrated via simulation. Let $\mathcal{H}_0$ denote the set of null hypotheses of the form $H_0^{\{\mathcal{G}_1, \mathcal{G}_2\}}$ that are true. The following property holds. 
\begin{prop} 
\label{prop:random1}
When $K = 2$, i.e. the clustering algorithm $\mathcal{C}(\cdot)$ estimates two clusters, 
\begin{align} \mathbb{P}\left (p({\bf X}; \mathcal{C}({\bf X})) \leq \alpha ~\Big | ~ H_0^{\mathcal{C}({\bf X})} \in \mathcal{H}_0  \right ) = \alpha, ~ \text{ for all } 0 \leq \alpha \leq 1, \label{eq:random1}
\end{align} 
where $p(\cdot; \cdot)$ is  defined in \eqref{eq:pval-def}. That is, if the two estimated clusters have the same mean, then the probability of falsely declaring otherwise is equal to $\alpha$. 
\end{prop}  
We prove Proposition \ref{prop:random1} in Appendix \ref{sec:proof-random1}.
What if $K > 2$? Then, given a data set ${\bf x}$, we can randomly select  (independently of ${\bf x}$) a single pair of estimated clusters $\mathcal{G}_1({\bf x}), \mathcal{G}_2({\bf x}) \in \mathcal{C}({\bf x})$. This leads to the following property.
\begin{prop} 
\label{prop:random2} 
Suppose that $K > 2$, i.e. the clustering algorithm $\mathcal{C}(\cdot)$ estimates three or more clusters, and let $\mathcal{G}_1({\bf x}), \mathcal{G}_2({\bf x}) \in \mathcal{C}({\bf x})$ denote a randomly selected pair. If $\{\mathcal{G}_1({\bf X}), \mathcal{G}_2({\bf X})\}$ and ${\bf X}$ are conditionally independent given $\mathcal{C}({\bf X})$, then for $p(\cdot; \cdot)$ defined in \eqref{eq:pval-def}, 
\begin{align} 
\mathbb{P}\left (p({\bf X};  \{\mathcal{G}_1({\bf X}), \mathcal{G}_2({\bf X}) \}) \leq \alpha  ~\Big | ~ H_0^{\left \{ \mathcal{G}_1({\bf X}), \mathcal{G}_2({\bf X}) \right \}}  \in \mathcal{H}_0 \right ) = \alpha, ~ \text{ for all } 0 \leq \alpha \leq 1. \label{eq:random2}
\end{align} 
\end{prop} 
We prove Proposition \ref{prop:random2} in Appendix \ref{sec:proof-random1}.
Recall that in Figure \ref{fig:pval-naive}(c), we simulated data with $\bm \mu = \bm 0_{n \times q}$, so the conditioning event in \eqref{eq:random2} holds with probability 1. Thus, \eqref{eq:random2} specializes to $\mathbb{P}(p({\bf X};  \{\mathcal{G}_1({\bf X}), \mathcal{G}_2({\bf X}) \}) \leq \alpha) = \alpha$, i.e. $p({\bf X};  \{\mathcal{G}_1({\bf X}), \mathcal{G}_2({\bf X}) \}) \sim \text{Uniform}(0, 1)$. This property is illustrated in Figure \ref{fig:pval-naive}(c). 

\section{Simulation results}
\label{sec:sim}

Throughout this section, we use the efficient implementation of hierarchical clustering in the \texttt{fastcluster} package \citep{mullner2013fastcluster} in \texttt{R}. 

\subsection{Uniform p-values under a global null} 
\label{sec:type1}
We generate data from \eqref{eq:mod} with ${\bm \mu} = \bm 0_{n \times q}$, so that  $H_0^{\{\hc_1, \hc_2\}}$ holds for all pairs of estimated clusters.  We simulate 2000 data sets for $n = 150$, $\sigma \in \{1, 2, 10\}$, and $q \in \{2, 10, 100\}$. For each data set, we use average, centroid, single, and complete linkage hierarchical clustering to estimate three clusters, and then test $H_0^{\{\hc_1, \hc_2\}}$ for a randomly-chosen pair of clusters. We compute the p-value defined in \eqref{eq:pval-def} as described in Section \ref{sec:hier} for average, centroid, and single linkage. For complete linkage, we approximate the p-value as described in Section \ref{sec:approx} with $N = 2000$. Figure \ref{fig:type1} displays QQ plots of the empirical distribution of the p-values against the Uniform$(0, 1)$ distribution. The p-values have a Uniform$(0, 1)$ distribution because our proposed test satisfies \eqref{eq:random2} and because $\bm \mu =\bm 0_{n \times q}$; see the end of Section \ref{sec:consequence} for a detailed discussion. In Appendix \ref{sec:supp-sim}, we show that plugging in an over-estimate $\sigma$ as in \eqref{eq:est-pval} yields p-values that are stochastically larger than the Uniform$(0, 1)$ distribution. 

\begin{figure}[h!]
\centering
\includegraphics[scale=0.5]{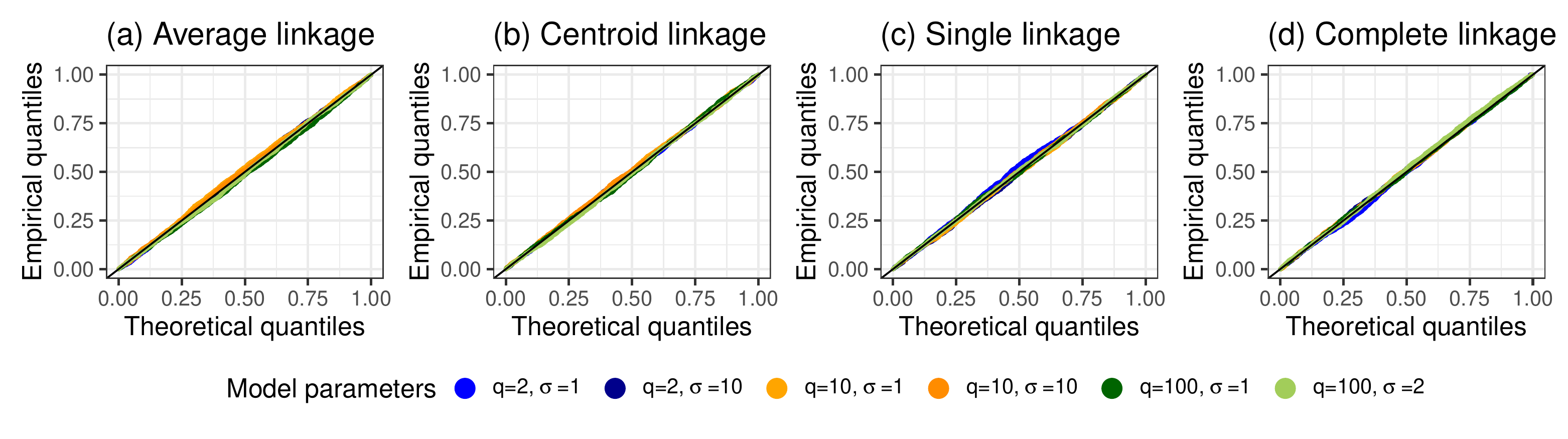}
\caption{\label{fig:type1} For 2000 draws from \eqref{eq:mod} with $\bm \mu = \bm 0_{n \times q}$, $n = 150$, $q \in \{2, 10, 100\}$, and $\sigma \in \{1,  2,10\}$, QQ-plots of the p-values obtained from the test proposed in Section \ref{sec:frame-test}, using (a) average linkage, (b) centroid linkage, (c) single linkage, and (d) complete linkage.}
\end{figure}

\subsection{Conditional power and recovery probability} 
\label{sec:power} 
\label{sec:cond} 
We generate data from \eqref{eq:mod} with $n = 30$, and three equidistant clusters,
\begin{align} 
\mu_1 = \cdots = \mu_{\frac{n}{3}} = \left [ \begin{smallmatrix} -\delta/2 \\ 0_{q-1} \end{smallmatrix} \right ] ,  \mu_{\frac{n}{3}+1} = \cdots = \mu_{\frac{2n}{3}} = \left [ \begin{smallmatrix} 0_{q-1} \\ \sqrt{3}\delta/2 \end{smallmatrix} \right ],   \mu_{\frac{2n}{3}+1} = \cdots = \mu_{n} =\left [ \begin{smallmatrix} \delta/2 \\ 0_{q-1} \end{smallmatrix} \right ], \label{eq:equi-means}
\end{align} 
for $\delta > 0$. For each simulated data set, we use average, centroid, single, and complete linkage hierarchical clustering to estimate three clusters, and then test $H_0^{\{\hc_1, \hc_2\}}$ for a randomly-chosen pair of clusters, with significance level $\alpha = 0.05$. We simulate $300,000$ data sets for $\sigma = 1$,  $q = 10$, and seven evenly-spaced values of $\delta \in [4, 7]$. {We define the \emph{conditional power} to be the conditional probability of rejecting $H_0^{\{\hc_1, \hc_2\}}$ for a randomly-chosen pair of estimated clusters, given that the randomly-chosen pair of estimated clusters correspond to two true clusters. Since this conditions on the event that the randomly-chosen pair of estimated clusters correspond to two true clusters, we are also interested in how often this event occurs. We therefore consider the \emph{recovery probability}, the probability that the randomly-chosen pair of estimated clusters correspond to two true clusters. Figure \ref{fig:cond} displays the conditional power and recovery probability as a function of the distance between the true clusters $(\delta$).

Figure \ref{fig:cond} displays conditional power and recovery probability as a function of the distance between the true clusters $(\delta$).
\begin{figure}[h!]
\centering 
\includegraphics[scale=0.4]{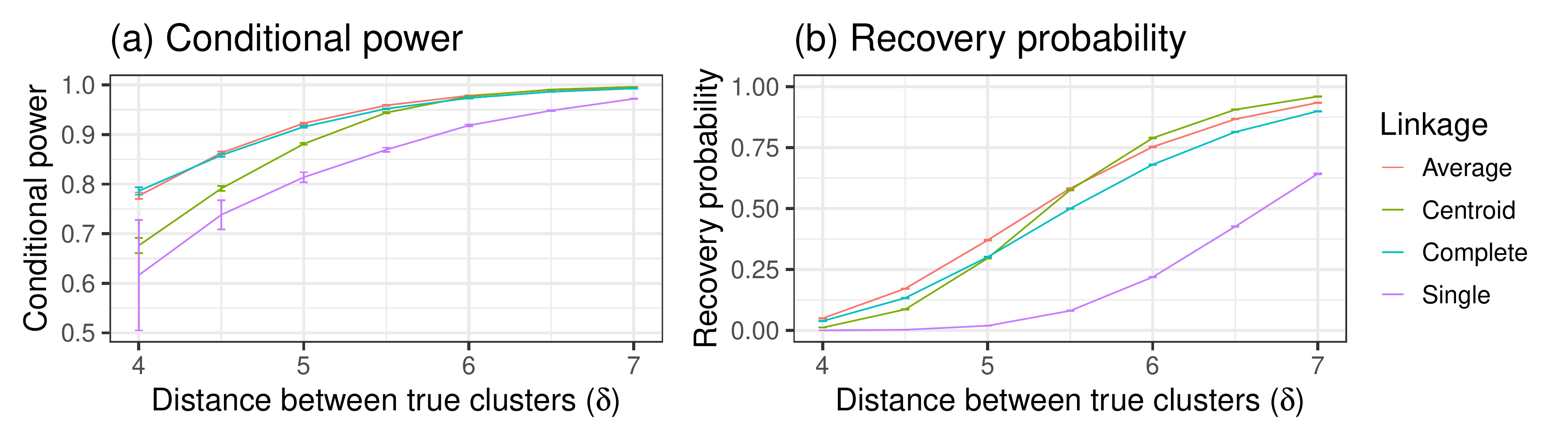}
\caption{\label{fig:cond} For the simulation study described in Section \ref{sec:cond}, (a) conditional power of the test proposed in Section \ref{sec:frame} versus the difference in means between the true clusters $(\delta$), and (b) recovery probability versus $\delta$. Conditional power and recovery probability are defined in Section \ref{sec:cond}.}
\end{figure} 
For all four linkages, the conditional power and recovery probability increase as the distance between the true clusters $(\delta)$ increases. Average and complete linkage have the highest conditional power, and single linkage has the lowest conditional power. Average, centroid, and complete linkage have substantially higher recovery probabilities than single linkage. 

We consider an alternative definition of power that does not condition on having correctly estimated the true clusters in Appendix \ref{sec:effect}. 

\section{Data applications}
\label{sec:app}
\subsection{Palmer penguins \citep{palmerpenguins}}
\label{sec:penguins}
In this section, we analyze the \verb+penguins+ data set from the \verb+palmerpenguins+ package in \verb+R+ \citep{palmerpenguins}. We estimate $\sigma$ with $\hat \sigma({\bf x}) = \sqrt{\sum \limits_{i=1}^n \sum \limits_{j=1}^q (x_i - \bar{x}_j)^2/(nq-q)}$ for $\bar{x}_j = \sum \limits_{i=1}^n x_{ij}/n$, calculated on 
the bill length and flipper length of 58 female penguins observed in the year 2009. We then consider the 107 female penguins observed in the years 2007--2008 with complete data on species, bill length, and flipper length. 
Figure \ref{fig:penguin-scatter}(a) displays the dendrogram obtained from applying average linkage hierarchical clustering with squared Euclidean distance to the penguins' bill length and flipper length, cut to yield five clusters, and Figure \ref{fig:penguin-scatter}(b) displays the data.

\begin{figure}[h!] 
\centering
\includegraphics[scale=0.45]{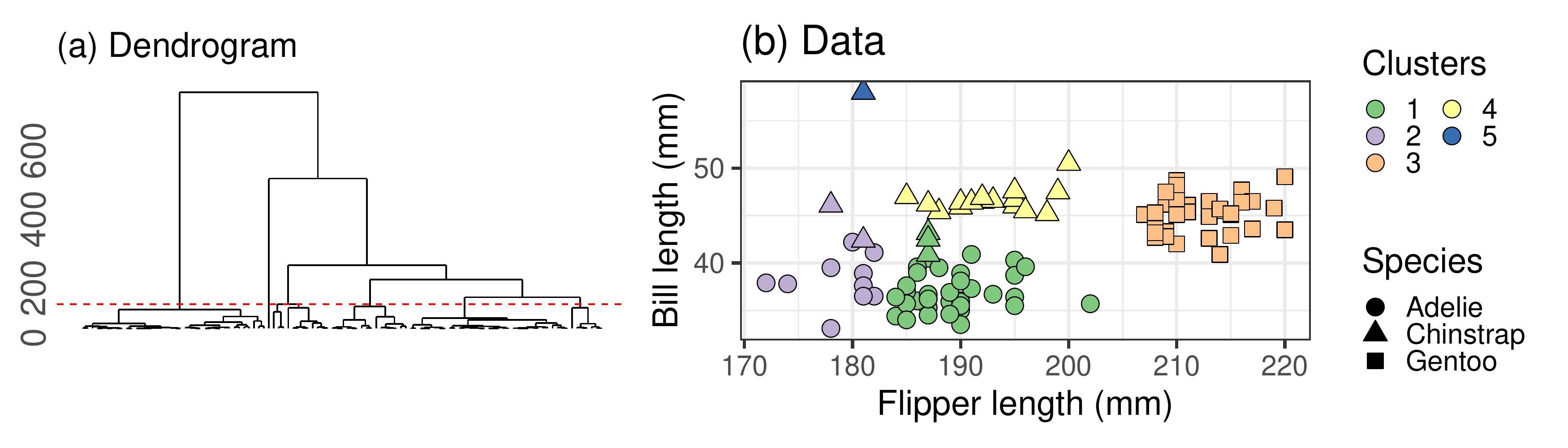}
\caption{\label{fig:penguin-scatter} (a) The average-linkage hierarchical clustering dendrogram and  (b) the bill lengths and flipper lengths, as well as the true species labels and estimated clusters, for the Palmer penguins data described in Section \ref{sec:penguins}. }
\end{figure} 

We  test $H_0^{\{\hc_1, \hc_2\}}$ for all pairs of clusters that contain more than one observation, using the test proposed in Section \ref{sec:frame-test}, and using the Wald test described in \eqref{eq:wald}. (The latter does not account for the fact that the clusters were estimated from the data, and does not control the selective type I error rate.)  Results are in Table \ref{tab:penguins}. The Wald p-values are small, even when testing for a difference in means between a single species (Clusters 1 and 2). Our proposed test yields a large p-value when testing for a difference in means between a single species (Clusters 1 and 2), and small p-values when the clusters correspond to different species (Clusters 1 and 3, and Clusters 3 and 4). The p-values from our proposed test are large for the remaining three pairs of clusters containing different species, even though visual inspection suggests a large difference between these two clusters. This is because the test statistic is close to the left boundary of $\mathcal{\hat S}$ defined in \eqref{eq:defS}, which leads to low power: see the discussion of Figure \ref{fig:boundary} in Appendix \ref{sec:effect}. 

\begin{table}[h!]
\centering
\begin{tabular}{lllllll}
Cluster pairs & $(1, 2)$ & $(1, 3)$ & ($1, 4)$ & $(2, 3$) & $(2, 4)$ & $(3, 4)$  \\  \hline
Test statistic        & $10.1$   & $25.0$   & $10.1$ & $33.8$ & $17.1$ & $18.9$  \\ 
Our p-value           & $0.591$    & $1.70 \times 10^{-14}$ &  $0.714$ & $0.070$ & $0.291$ & $2.10 \times 10^{-6} $      \\ 
Wald p-value          & $0.00383$  &  $< 10^{-307}$ & $0.00101$  &  $< 10^{-307}$ & $4.29 \times 10^{-5}$ & $1.58 \times 10^{-11}$\\ 
\end{tabular}
\caption{\label{tab:penguins} Results from applying the test of $H_0^{\{\hc_1, \hc_2\}}: \bar{\mu}_{\hc_k} = \bar{\mu}_{\hc_{k'}}$ proposed in Section \ref{sec:frame} and the Wald test defined in \eqref{eq:wald} to the Palmer penguins data set, displayed in Figure \ref{fig:penguin-scatter}(b). 
}
\end{table}

\subsection{Single-cell RNA sequencing data \citep{zheng2017massively}} 
\label{sec:single-cell}

Single-cell RNA sequencing data quantifies the gene expression levels of individual cells.
Biologists often cluster the  cells to identify putative cell types, and then test for differential gene expression between the clusters, without properly accounting for the fact that the clusters were estimated from the data \citep{luecken2019current, lahnemann2020eleven, deconinck2021recent}. \citet{zheng2017massively} classified peripheral blood mononuclear cells prior to sequencing. We will use this data set to demonstrate that testing for differential gene expression after clustering with our proposed selective inference framework yields reasonable results.
 
\subsubsection{Data and pre-processing} 
\label{sec:data} 
We subset the data to the memory T cells, B cells, and monocytes. In line with standard pre-processing protocols \citep{duo2018systematic}, we remove cells with a  high mitochondrial gene percentage, cells with a  low or  high number of expressed genes, and cells with a  low number of total counts. Then, we scale the data so that the total number of counts for each cell equals the average count across all cells. Finally, we apply a $\log_2$ transformation with a pseudo-count of 1, and subset to the 500 genes with the largest pre-normalization average expression levels. We separately apply this pre-processing routine to the memory T cells only, and to all of the cells. After pre-processing, we construct a ``no clusters" data set by randomly sampling 600 memory T cells, and  a ``clusters" data set by randomly sampling 200 each of memory T cells, B cells, and monocytes. 

\subsubsection{Data analysis} 
\label{sec:negative}
We apply Ward-linkage hierarchical clustering with squared Euclidean distance to the ``no clusters" data to get three clusters, containing $64$, $428$, and $108$ cells, respectively. For each pair of clusters, we test $H_0^{\{\hc_1, \hc_2\}}: \bar{\mu}_{\hc_1} = \bar{\mu}_{\hc_2}$ using (i) the test proposed in Section \ref{sec:general-cov} under model \eqref{eq:mod2} and (ii) using the Wald test under model \eqref{eq:mod2}, which has p-value 
\begin{align} 
\mathbb{P}_{H_0^{\{\hc_1, \hc_2\}}} \left ( (\bar{X}_{\hc_1} - \bar{X}_{\hc_2})^T \bm \Sigma^{-1}(\bar{X}_{\hc_1} - \bar{X}_{\hc_2})\geq (\bar{x}_{\hc_1} - \bar{x}_{\hc_2})^T \bm \Sigma^{-1}(\bar{x}_{\hc_1} - \bar{x}_{\hc_2}) \right ). \label{eq:wald2}
\end{align} 
For both tests, we estimate $\bm \Sigma$ by applying the principal complement thresholding  (``POET", \citealt{fan2013large}) method to the 
9,303 memory T cells left out of the ``no clusters" data set. Results are in Table \ref{tab:negative}. The p-values from our test are large, and the Wald p-values are small. Recall that all of the cells are memory T cells, and so (as far as we know) there are no true clusters in the data. 

\begin{table}[h!]
\centering
\begin{tabular}{l|lll|lll}
               & \multicolumn{3}{c|}{``No clusters"} & \multicolumn{3}{c}{``Clusters"l} \\ \cline{2-7} 
Cluster pairs  & (1, 2)     & (1, 3)     & (2, 3)      & (1, 2)     & (1, 3)     & (2, 3)     \\ \hline
Test statistic & 4.05       & 4.76       & 2.96        & 3.04       & 4.27       & 4.38       \\
Our p-value    & 0.20       & 0.27       & 0.70        & $4.60 \times 10^{-28} $    & $3.20 \times 10^{-82} $    & $1.13 \times 10^{-73} $    \\
Wald p-value   & $< 10^{-307}$          & $< 10^{-307}$          & $< 10^{-307}$    &  $< 10^{-307}$      &  $< 10^{-307}$      & $< 10^{-307}$ 
\end{tabular} 
\caption{\label{tab:negative} Results from applying the test of $H_0^{\{\hc_1, \hc_2\}}: \bar{\mu}_{\hc_1} = \bar{\mu}_{\hc_2}$ proposed in Section \ref{sec:general-cov} and the Wald test in \eqref{eq:wald2} to the ``no clusters" and ``clusters" data described in Section \ref{sec:data}. }
\end{table}

We now apply the same analysis to the ``clusters" data. Ward-linkage hierarchical clustering with squared Euclidean distance results in three clusters that almost exactly correspond to memory T cells, B cells, and monocytes. For both tests, we estimate $\bm \Sigma$ by applying the POET method to the 21,757 memory T cells, B cells, and monocytes left out of the ``clusters" data set. Results are in Table \ref{tab:negative}. The p-values from both tests are extremely small. This suggests that our proposed approach is able to correctly reject the null hypothesis when it does not hold.

\section{Discussion}
\label{sec:discuss}
In this paper, we proposed a selective inference framework for testing the null hypothesis that there is no difference in means between two estimated clusters, under \eqref{eq:mod}. This directly solves a problem that routinely occurs when biologists analyze data, e.g. when testing for differential expression between clusters estimated from single-cell RNA-sequencing data \citep{luecken2019current, lahnemann2020eleven, deconinck2021recent}. 

The framework proposed in Section \ref{sec:frame} assumed that $\sigma$ in \eqref{eq:mod} was known. Similarly, the extended framework in Section \ref{sec:general-cov} assumed that $\bm \Sigma$ in \eqref{eq:mod2} was known. These assumptions are unlikely to hold in practice. Thus, in the data applications of Sections \ref{sec:penguins}--\ref{sec:single-cell}, we replaced $\sigma$ in \eqref{eq:mod} and $\bm \Sigma$ in \eqref{eq:mod2} with estimates. In Section \ref{sec:unknown}, we explored the theoretical implications of this replacement, under model \eqref{eq:mod}. Future work could include investigating how best to estimate $\sigma$ in \eqref{eq:mod} and $\bm \Sigma$ in \eqref{eq:mod2}. Another potential avenue for future work would be to develop an analogue of Theorem \ref{thm:plug-in-conservative} in Section \ref{sec:unknown} under model \eqref{eq:mod2}.

The tests developed in this paper are implemented in the \verb+R+ package \verb+clusterpval+. Instructions on how to download and use this package can be found at\\  \verb+http://lucylgao.com/clusterpval+. Links to download the data sets in Section \ref{sec:app} can be found at \verb+https://github.com/lucylgao/clusterpval-experiments+, along with code to reproduce the simulation and real data analysis results from this paper.

\section*{Acknowledgments}
Lucy L. Gao was supported by the NSERC Discovery Grants program.  Daniela Witten and Jacob Bien were supported by NIH Grant R01-GM123993.  Jacob Bien was supported by NSF CAREER Award DMS-1653017. Daniela Witten was supported by NSF CAREER Award DMS-1252624, NIH Grants R01-EB026908 and R01-DA047869, and an Simons Investigator Award in Mathematical Modeling of Living Systems (No. 560585). 

\section*{Conflict of interest statement}
The authors have no relevant financial or non-financial competing interests to report.

\bibliography{refs}

\appendix

\numberwithin{equation}{section}
\numberwithin{lemma}{section}
\numberwithin{figure}{section}

\setcounter{equation}{0}

\section{Proofs} 

\subsection{Proof of Theorem \ref{thm:pval}}
\label{sec:proof-pval}
The proof of Theorem \ref{thm:pval} is similar to the proof of Theorem 3.1 in \citet{loftus2015selective}, the proof of Lemma 1 in \citet{yang2016selective}, and the proof of Theorem 3.1 in  \citet{chen2019valid}. 

For any $\nu \in \mathbb{R}^n$, we have 
\begin{align} 
{\bf X} = \bm \pi_{\nu}^\perp {\bf X} + ({\bf I}_n - \bm \pi_{\nu}^\perp) {\bf X} = \bm \pi_{\nu}^\perp {\bf X}  +  \left ( \frac{\| {\bf X}^T \nu \|_2}{\| \nu\|_2^2}\right ) \nu  \text{dir}({\bf X}^T  \nu)^T.
 \label{eq:decomp}
\end{align} 
Therefore,
\begin{align} 
{\bf X} &=  \bm \pi_{\nu(\mathcal{G}_1, \mathcal{G}_2)}^\perp {\bf X}  +  \left ( \frac{\| {\bf X}^T \nu(\mathcal{G}_1, \mathcal{G}_2)\|_2}{\| \nu(\mathcal{G}_1, \mathcal{G}_2)\|_2^2}\right ) \nu ~ \text{dir}({\bf X}^T \nu(\mathcal{G}_1, \mathcal{G}_2))^T \nonumber \\ 
&= \bm \pi_{\nu(\mathcal{G}_1, \mathcal{G}_2)}^\perp {\bf X} +  \left ( \frac{\|\bar{X}_{\mathcal{G}_1} - \bar{X}_{\mathcal{G}_2}\|_2}{\frac{1}{|\mathcal{G}_1|} + \frac{1}{|\mathcal{G}_2|}}  \right ) \nu(\mathcal{G}_1, \mathcal{G}_2) ~ \text{dir}(\bar{X}_{\mathcal{G}_1} - \bar{X}_{\mathcal{G}_2})^T, \label{eq:decomp2} 
\end{align} 
where the first equality follows from \eqref{eq:decomp}, and the second equality follows from the definition of $\nu(\mathcal{G}_1, \mathcal{G}_2)$ in \eqref{eq:def-nu}. Substituting \eqref{eq:decomp2} into the definition of $p({\bf x}; \{\mathcal{G}_1, \mathcal{G}_2\})$ given by \eqref{eq:pval-def} yields
\begin{align} 
& p({\bf x}; \{\mathcal{G}_1, \mathcal{G}_2\}) \label{eq:continue} \\ 
&= \mathbb{P}_{H_0^{\{\mathcal{G}_1, \mathcal{G}_2\}}} \Bigg ( \|\bar{X}_{\mathcal{G}_1} - \bar{X}_{\mathcal{G}_2}\|_2 \geq \|\bar{x}_{\mathcal{G}_1} - \bar{x}_{\mathcal{G}_2}\|_2 ~ \bigg | ~ \nonumber \\
&\hspace{25mm} \mathcal{G}_1, \mathcal{G}_2 \in \mathcal{C} \left (  \bm \pi_{\nu(\mathcal{G}_1, \mathcal{G}_2)}^\perp {\bf x} +  \left ( \frac{\|\bar{X}_{\mathcal{G}_1} - \bar{X}_{\mathcal{G}_2}\|_2}{\frac{1}{|\mathcal{G}_1|} + \frac{1}{|\mathcal{G}_2|}}  \right ) \nu(\mathcal{G}_1, \mathcal{G}_2) \text{dir}(\bar{x}_{\mathcal{G}_1} - \bar{x}_{\mathcal{G}_2})^T \right ), \nonumber \\ 
&\hspace{25mm} \bm \pi_{\nu(\mathcal{G}_1, \mathcal{G}_2)}^\perp {\bf X} = \bm \pi_{\nu(\mathcal{G}_1, \mathcal{G}_2)}^\perp {\bf x}, ~\text{dir}(\bar{X}_{\mathcal{G}_1} - \bar{X}_{\mathcal{G}_2}) = \text{dir}(\bar{x}_{\mathcal{G}_1} - \bar{x}_{\mathcal{G}_2}) \Bigg ). \nonumber
\end{align} 
To simplify \eqref{eq:continue}, we now show that 
\begin{align} 
\|\bar{X}_{\mathcal{G}_1} - \bar{X}_{\mathcal{G}_2}\|_2 ~~\indep ~~\bm \pi_{\nu(\mathcal{G}_1, \mathcal{G}_2)}^\perp {\bf X},\label{eq:ind1} 
\end{align} 
and that under $H_0^{\{\mathcal{G}_1, \mathcal{G}_2\}}: \bar{\mu}_{\mathcal{G}_1} = \bar{\mu}_{\mathcal{G}_2}$, 
\begin{align} 
\|\bar{X}_{\mathcal{G}_1} - \bar{X}_{\mathcal{G}_2}\|_2~~ \indep ~~ \text{dir}(\bar{X}_{\mathcal{G}_1} - \bar{X}_{\mathcal{G}_2}). \label{eq:ind2} 
\end{align} 
 First, we will show \eqref{eq:ind1}. Recall that $\bm \pi_{\nu}^\perp$ is the orthogonal projection matrix onto the subspace orthogonal to $\nu$. Thus, $\bm \pi_{\nu(\mathcal{G}_1, \mathcal{G}_2)}^\perp \nu(\mathcal{G}_1, \mathcal{G}_2) = 0_n$. It follows from properties of the matrix normal and multivariate normal distributions that $\bm \pi_{\nu(\mathcal{G}_1, \mathcal{G}_2)}^\perp {\bf X} ~~\indep~~ {\bf X}^T \nu(\mathcal{G}_1, \mathcal{G}_2)$. This implies \eqref{eq:ind1}, since $\bar{X}_{\mathcal{G}_1} - \bar{X}_{\mathcal{G}_2} = {\bf X}^T \nu(\mathcal{G}_1, \mathcal{G}_2)$. Second, we will show \eqref{eq:ind2}. It follows from \eqref{eq:mod} that ${\bf X}_i \overset{ind}{\sim} N_q(\mu_i, \sigma^2 {\bf I}_q)$. Thus, under $H_0^{\{\mathcal{G}_1, \mathcal{G}_2\}}: \bar{\mu}_{\mathcal{G}_1} = \bar{\mu}_{\mathcal{G}_2}$,   
\begin{align} \frac{\bar{X}_{\mathcal{G}_1} - \bar{X}_{\mathcal{G}_2}}{\sigma \sqrt{1/|\mathcal{G}_1| + 1/|\mathcal{G}_2|}} = \frac{{\bf X}^T \nu(\mathcal{G}_1, \mathcal{G}_2)}{\|\nu(\mathcal{G}_1, \mathcal{G}_2)\|_2} \sim N_q\left(0 , {\bf I}_q \right ), \label{eq:normal} 
\end{align} 
and \eqref{eq:ind2} follows from the independence of the length and direction of a standard multivariate normal random vector. 

We now apply \eqref{eq:ind1} and \eqref{eq:ind2} to \eqref{eq:continue}. This yields 
 \begin{align} 
& \Scale[0.9]{ p({\bf x}; \{\mathcal{G}_1, \mathcal{G}_2\}) }\nonumber \\ 
&=  \Scale[0.9]{\mathbb{P}_{H_0^{\{\mathcal{G}_1, \mathcal{G}_2\}}} \Bigg (\|\bar{X}_{\mathcal{G}_1} - \bar{X}_{\mathcal{G}_2}\|_2 \geq \|\bar{x}_{\mathcal{G}_1} - \bar{x}_{\mathcal{G}_2}\|_2   ~\bigg | ~ \mathcal{G}_1, \mathcal{G}_2 \in  \mathcal{C} \Bigg (  \bm \pi_{\nu(\mathcal{G}_1, \mathcal{G}_2)}^\perp {\bf x} +  \Bigg ( \frac{\|\bar{X}_{\mathcal{G}_1} - \bar{X}_{\mathcal{G}_2}\|_2}{\frac{1}{|\mathcal{G}_1|} + \frac{1}{|\mathcal{G}_2|}}  \Bigg ) \nu(\mathcal{G}_1, \mathcal{G}_2) \text{dir}(\bar{x}_{\mathcal{G}_1} - \bar{x}_{\mathcal{G}_2})^T \Bigg ) \Bigg ) } \nonumber \\ 
&= \Scale[0.9]{\mathbb{P}_{H_0^{\{\mathcal{G}_1, \mathcal{G}_2\}}} \left (\|\bar{X}_{\mathcal{G}_1} - \bar{X}_{\mathcal{G}_2}\|_2 \geq \|\bar{x}_{\mathcal{G}_1} - \bar{x}_{\mathcal{G}_2}\|_2   ~\bigg | ~ \|\bar{X}_{\mathcal{G}_1} - \bar{X}_{\mathcal{G}_2}\|_2 \in \mathcal{S}({\bf x}; \{\mathcal{G}_1, \mathcal{G}_2\}) \right )}, \label{eq:continue2} 
\end{align} 
where \eqref{eq:continue2} follows from the definition of $\mathcal{S}({\bf x}; \{\mathcal{G}_1, \mathcal{G}_2\})$ in \eqref{eq:defS-general}. Under $H_0^{\{\mathcal{G}_1, \mathcal{G}_2\}} : \bar{\mu}_{\mathcal{G}_1} = \bar{\mu}_{\mathcal{G}_2} $, 
\begin{align*} \frac{\| \bar{X}_{\mathcal{G}_1} - \bar{X}_{\mathcal{G}_2}\|_2^2}{\sigma^2 \left ( 1/|\mathcal{G}_1| + 1/|\mathcal{G}_2| \right )} \sim \chi^2_q, 
\end{align*} 
by \eqref{eq:normal}. 
That is, under $H_0^{\{\mathcal{G}_1, \mathcal{G}_2\}}$, we have $\|\bar{X}_{\mathcal{G}_1} - \bar{X}_{\mathcal{G}_2}\|_2 \sim \left ( \sigma  \sqrt{1/|\mathcal{G}_1| + 1/|\mathcal{G}_2|} \right ) \cdot \chi_q$. Therefore, for $\phi \sim \left ( \sigma  \sqrt{1/|\mathcal{G}_1| + 1/|\mathcal{G}_2|} \right ) \cdot \chi_q$, it follows from \eqref{eq:continue2} that 
\begin{align} 
p({\bf x}; \{\mathcal{G}_1, \mathcal{G}_2\}) &= \mathbb{P}\left (\phi \geq \|\bar{x}_{\mathcal{G}_1} - \bar{x}_{\mathcal{G}_2}\|_2  \mid \phi \in \mathcal{S}({\bf x}; \{\mathcal{G}_1, \mathcal{G}_2\} ) \right ) \nonumber  \\ 
&= 1 - \mathbb{F}\left (\|\bar{x}_{\mathcal{G}_1} - \bar{x}_{\mathcal{G}_2}\|_2; \sigma  \sqrt{1/|\mathcal{G}_1| + 1/|\mathcal{G}_2|},  \mathcal{S}({\bf x}; \{\mathcal{G}_1, \mathcal{G}_2\} )  \right ), \label{eq:thing}
\end{align} 
since $\mathbb{F}(t; c, \mathcal{S})$ denotes the cumulative distribution function of a $c \cdot \chi_q$ random variable truncated to the set $\mathcal{S}$. This is \eqref{eq:pval-phi}. 

Finally, we will show \eqref{eq:selective-ours}. It follows from the definition of $p(\cdot; \{\mathcal{G}_1, \mathcal{G}_2\})$ in \eqref{eq:pval-def} that for all ${\bf x} \in \mathbb{R}^{n \times q}$, we have 
\begin{align} 
&\mathbb{P}_{H_0^{\{\mathcal{G}_1, \mathcal{G}_2\}}} \Big( p({\bf X}; \{\mathcal{G}_1, \mathcal{G}_2\}) \leq \alpha \mid \mathcal{G}_1, \mathcal{G}_2 \in \mathcal{C}({\bf X}), \bm \pi_{\nu(\mathcal{G}_1, \mathcal{G}_2)}^\perp {\bf X} = \bm \pi_{\nu(\mathcal{G}_1, \mathcal{G}_2)}^\perp {\bf x}, \nonumber\\ 
&\hspace{56mm} \text{dir}(\bar{X}_{\mathcal{G}_1} - \bar{X}_{\mathcal{G}_2}) = \text{dir}(\bar{x}_{\mathcal{G}_1} - \bar{x}_{\mathcal{G}_2})  \Big ) = \alpha. \label{eq:tower-start} 
\end{align} 
Therefore, letting $\mathbb{E}_{H_0}$ denote $\mathbb{E}_{H_0^{\{\mathcal{G}_1, \mathcal{G}_2\}}}$, we have
\begin{align*} 
&\mathbb{P}_{H_0^{\{\mathcal{G}_1, \mathcal{G}_2\}}}  \Big ( p({\bf X}; \{\mathcal{G}_1, \mathcal{G}_2\}) \leq \alpha ~ \Big |~  \mathcal{G}_1, \mathcal{G}_2 \in \mathcal{C}({\bf X}) \Big ) \\ 
&= \mathbb{E}_{H_0} \Big [ \mathds{1} \{ p({\bf X}; \{\mathcal{G}_1, \mathcal{G}_2\}) \leq \alpha \} ~\Big | ~   \mathcal{G}_1, \mathcal{G}_2 \in \mathcal{C}({\bf X}) \Big ]  \\ 
&= \mathbb{E}_{H_0} \Bigg [ \mathbb{E}_{H_0} \Big [ \mathds{1} \{ p({\bf X}; \{\mathcal{G}_1, \mathcal{G}_2\}) \leq \alpha \} ~ \Big | ~  \mathcal{G}_1, \mathcal{G}_2 \in \mathcal{C}({\bf X}), \bm \pi_{\nu(\mathcal{G}_1, \mathcal{G}_2)}^\perp {\bf X}, \text{dir}(\bar{X}_{\mathcal{G}_1} - \bar{X}_{\mathcal{G}_2})  \Big ]  \bigg | ~ \mathcal{G}_1, \mathcal{G}_2 \in \mathcal{C}({\bf X})  \Bigg]  \\ 
&= \mathbb{E}_{H_0} \Big [\alpha ~ \Big | ~ \mathcal{G}_1, \mathcal{G}_2 \in \mathcal{C}({\bf X}) \Big ]  \\ 
&= \alpha, 
\end{align*} 
where the second equality follows from the law of total expectation, and the third equality follows from \eqref{eq:tower-start}. This is \eqref{eq:selective-ours}. 

\subsection{Proof of Lemma~\ref{lem:heights}}
\label{sec:proof-hier}
First, we state and prove a preliminary result involving the dissimilarities between groups of observations in any perturbed data set ${\bf x}'(\phi)$, which holds for any clustering method $\mathcal{C}$. 
\begin{lemma} 
\label{lem:dissimilar} 
For any $\phi \geq 0$, and for any two sets $\mathcal{G}$ and $\mathcal{G}'$ that are both contained in $\hc_1$, both contained in $\hc_2$, or both contained in $\{1, 2, \ldots, n \} \backslash [\hc_1 \cup \hc_2]$, we have that
\begin{gather*} 
d(\mathcal{G}, \mathcal{G}'; {\bf x}'(\phi)) = d(\mathcal{G}, \mathcal{G}'; {\bf x}),
\end{gather*} 
where ${\bf x}'(\phi)$ is defined in  \eqref{eq:xphi}. 
\end{lemma}  
\begin{proof}
By \eqref{eq:xphi}, $[{\bf x}'(\phi)]_i - [{\bf x}'(\phi)]_{i'} = x_i - x_{i'}$ for all $i, i' \in \mathcal{G} \cup \mathcal{G}'$, and therefore all pairwise distances are preserved for all $i, i' \in \mathcal{G} \cup \mathcal{G}'$. The resut follows. 
\end{proof}

Suppose that $\mathcal{C} = \mathcal{C}^{(n-K+1)}$. Then, \eqref{eq:heights} follows immediately from observing that $\mathcal{W}_1^{(t)}({\bf x})$ and $\mathcal{W}_2^{(t)}({\bf x})$ are clusters merged at an earlier stage of the hierarchical clustering algorithm and thus must satisfy the condition in Lemma \ref{lem:dissimilar}. 

We will now show \eqref{eq:merges}. Applying the right-hand side of \eqref{eq:merges} with $t = n-K+1$ yields the left-hand side, so the $(\Leftarrow)$ direction holds. We will now prove the $(\Rightarrow)$ direction by contradiction. Suppose that there exists some $t \in \{1, \ldots, n-K-1\}$ such that 
\begin{gather} 
\mathcal{C}^{(t+1)}({\bf x}'(\phi)) \neq \mathcal{C}^{(t +1)}({\bf x}), \label{eq:induct-assume} 
\end{gather}
and let $t$ specifically denote the smallest such value. Then, $\{ \mathcal{W}_1^{(t)}({\bf x}'(\phi)), \mathcal{W}_2^{(t)}({\bf x}'(\phi)) \} \neq \{ \mathcal{W}_1^{(t)}({\bf x}), \mathcal{W}_2^{(t)}({\bf x}) \} $, which implies that  
\begin{align} 
d\left (\mathcal{W}_1^{(t)}({\bf x}), \mathcal{W}_2^{(t)}({\bf x} ); {\bf x}'(\phi) \right ) > d\left (\mathcal{W}_1^{(t)}({\bf x}'(\phi)), \mathcal{W}_2^{(t)}({\bf x}'(\phi)); {\bf x}'(\phi) \right ). \label{eq:contra} 
\end{align} 
Since clusters cannot unmerge once they have merged, and $\hc_1, \hc_2 \in \mathcal{C}^{(n-K+1)}({\bf x})$ by definition, it must be the case that $\mathcal{W}_1^{(t)}({\bf x})$ and $\mathcal{W}_2^{(t)}({\bf x})$ are both in $\hc_1$, are both in $\hc_2$, or are both in $[\hc_1 \cup \hc_2]^C$. Similarly, since we assumed that $\hc_1, \hc_2 \in \mathcal{C}^{(n-K+1)}({\bf x}'(\phi))$, it must be the case that $\mathcal{W}_1^{(t)}({\bf x}'(\phi))$ and $\mathcal{W}_2^{(t)}({\bf x}'(\phi))$ are  both in $\hc_1$, are both in $\hc_2$, or are both in $[\hc_1 \cup \hc_2]^C$. Thus, we can apply Lemma \ref{lem:dissimilar} to both sides of \eqref{eq:contra} to yield 
\begin{align} 
d\left (\mathcal{W}_1^{(t)}({\bf x}), \mathcal{W}_2^{(t)}({\bf x} ); {\bf x}\right ) > d\left (\mathcal{W}_1^{(t)}({\bf x}'(\phi)), \mathcal{W}_2^{(t)}({\bf x}'(\phi)); {\bf x} \right ), 
\end{align} 
which implies that $\left \{ \mathcal{W}_1^{(t)}({\bf x}), \mathcal{W}_2^{(t)}({\bf x}) \right \}$ is not the pair of clusters that merged in the $(t)^{th}$ step of the hierarchical clustering procedure applied to ${\bf x}$. This is a contradiction. Therefore, for all $t \in \{1, \ldots, n-K-1\}$ such that $\mathcal{C}^{(t)}({\bf x}'(\phi)) = \mathcal{C}^{(t)}({\bf x})$, it must be the case that $\mathcal{C}^{(t+1)}({\bf x}'(\phi)) = \mathcal{C}^{(t+1)}({\bf x})$. Observing that $\mathcal{C}^{(1)}({\bf x}'(\phi)) = \mathcal{C}^{(1)}({\bf x}) = \{\{1\}, \{2\}, \ldots, \{n\} \}$ for all $\phi \geq 0$ completes the proof of the  $(\Rightarrow)$ direction.

\subsection{Proof of Theorem \ref{thm:hierS}}
\label{sec:proofS}
Observe from Algorithm \ref{alg:hier} that 
 \begin{align} \mathcal{C}^{(t)}({\bf x}'(\phi)) = \mathcal{C}^{(t)}({\bf x}) \text{ for all } t = 1, 2, \ldots, n-K+1, \label{eq:heights1} 
 \end{align} 
if and only if for all $ t = 1, 2, \ldots, n-K$, 
\begin{gather} 
d(\mathcal{G},\mathcal{G}'; {\bf x}'(\phi)) > d \left (\mathcal{W}_1^{(t)}({\bf x}), \mathcal{W}_2^{(t)}({\bf x}); {\bf x}'(\phi) \right ), \label{eq:merges1} \\ 
\forall ~ \mathcal{G}, \mathcal{G}' \in \mathcal{C}^{(t)}({\bf x}), ~ \mathcal{G} \neq \mathcal{G'}, ~ \{\mathcal{G}, \mathcal{G}'\} \neq \left  \{\mathcal{W}_1^{(t)}({\bf x}), \mathcal{W}_2^{(t)}({\bf x}) \right \}.  \nonumber 
\end{gather}
Equation \eqref{eq:heights} in Lemma \ref{lem:heights} says that 
\begin{gather*} 
d \left (\mathcal{W}_1^{(t)}({\bf x}), \mathcal{W}_2^{(t)}({\bf x}); {\bf x}'(\phi) \right ) = d \left (\mathcal{W}_1^{(t)}({\bf x}), \mathcal{W}_2^{(t)}({\bf x}); {\bf x} \right ), ~\forall~ t = 1, 2, \ldots, n-K, ~ \forall ~\phi \geq 0. 
\end{gather*} 
Therefore, \eqref{eq:heights1} is true if and only if  for all $ t = 1, 2, \ldots, n-K$, 
\begin{gather} 
d(\mathcal{G},\mathcal{G}'; {\bf x}'(\phi)) > d \left (\mathcal{W}_1^{(t)}({\bf x}), \mathcal{W}_2^{(t)}({\bf x}); {\bf x} \right ),  \label{eq:l-indexing}\\ 
\forall ~ \mathcal{G}, \mathcal{G}' \in \mathcal{C}^{(t)}({\bf x}), ~ \mathcal{G} \neq \mathcal{G'}, ~ \{\mathcal{G}, \mathcal{G}'\} \neq \left  \{\mathcal{W}_1^{(t)}({\bf x}), \mathcal{W}_2^{(t)}({\bf x}) \right \}.  \nonumber  
\end{gather}
Recall that $\mathcal{L}({\bf x}) = \bigcup \limits_{t=1}^{n-K} \left \{ \{\mathcal{G}, \mathcal{G}'\}: \mathcal{G}, \mathcal{G}' \in \mathcal{C}^{(t)}({\bf x}), \mathcal{G} \neq \mathcal{G}', \{ \mathcal{G}, \mathcal{G}' \} \neq \left \{ \mathcal{W}_1^{(t)}({\bf x}), \mathcal{W}_2^{(t)}({\bf x}) \right \} \right \}$.  Observe that \eqref{eq:l-indexing} is true for all $t = 1, 2, \ldots, n-K$ if and only if for all $\{\mathcal{G}, \mathcal{G}'\} \in \mathcal{L}({\bf x})$, 
\begin{align} 
d(\mathcal{G},\mathcal{G}'; {\bf x}'(\phi)) > \underset{l_{\mathcal{G}, \mathcal{G}'}({\bf x}) \leq t \leq u_{\mathcal{G}, \mathcal{G}'}({\bf x})}{\max}  d\left ( \mathcal{W}_1^{(t)}({\bf x}), \mathcal{W}_2^{(t)}({\bf x}); {\bf x} \right ), \label{eq:pair-indexing}  
\end{align}
where $l_{\mathcal{G}, \mathcal{G}'}({\bf x}) =  \min \left \{1 \leq t \leq n-K: \mathcal{G}, \mathcal{G}' \in \mathcal{C}^{(t)}({\bf x}),  \{\mathcal{G}, \mathcal{G}' \} \neq \left \{ \mathcal{W}_1^{(t)}({\bf x}), \mathcal{W}_2^{(t)}({\bf x})\}  \right \}  \right \}$ and
$u_{\mathcal{G}, \mathcal{G}'}({\bf x}) =  \max \left \{ 1 \leq t \leq n-K: \mathcal{G}, \mathcal{G}' \in \mathcal{C}^{(t)}({\bf x}), \{\mathcal{G}, \mathcal{G}' \} \neq \{ \mathcal{W}_1^{(t)}({\bf x}), \mathcal{W}_2^{(t)}({\bf x})\}  \right \}  $ are the start and end of the lifetime of the cluster pair $\{\mathcal{G}, \mathcal{G}'\}$, respectively.
Therefore, it follows from \eqref{eq:merges} in Lemma \ref{lem:heights} that $\hc_1, \hc_2 \in \mathcal{C}({\bf x}'(\phi))$ if and only if \eqref{eq:pair-indexing} is true for all $\{\mathcal{G}, \mathcal{G}'\} \in \mathcal{L}({\bf x})$. Recalling from \eqref{eq:defS} that $\mathcal{\hat S} = \{ \phi \geq 0: \hc_1, \hc_2 \in  \mathcal{C}({\bf x}'(\phi))\}$, it follows that 
\begin{align*} 
\mathcal{\hat S} = \bigcap \limits_{\{\mathcal{G}, \mathcal{G}'\} \in \mathcal{L}({\bf x})} \left \{ \phi \geq 0: d(\mathcal{G}, \mathcal{G}'; {\bf x}'(\phi)) > \underset{l_{\mathcal{G},\mathcal{G}'}({\bf x}) \leq t \leq u_{\mathcal{G},\mathcal{G}'}({\bf x})}{\max} d\left ( \mathcal{W}_1^{(t)}({\bf x}),\mathcal{W}_2^{(t)}({\bf x})  ; {\bf x} \right ) \right \}. 
\end{align*}
This is \eqref{eq:hierS}. Furthermore, this intersection involves $\mathcal{O}(n^2)$ sets, because 
\begin{align*} 
|\mathcal{L}({\bf x})| = \left ( \binom{n}{2} - 1 \right ) + \sum \limits_{t=1}^{n-K-1} (n-t-1). 
\end{align*}

\subsection{Proof of Proposition \ref{lem:computeH}}
\label{sec:proof-computeH}
Let $\{\mathcal{G}, \mathcal{G}' \} \in \mathcal{L}({\bf x})$, where $\mathcal{L}({\bf x})$ in \eqref{eq:losers} is the set of losing cluster pairs  in ${\bf x}$. Suppose that we are given the start and end of the lifetime of this cluster pair in ${\bf x}$, $l_{\mathcal{G}, \mathcal{G}'}({\bf x})$  and $u_{\mathcal{G}, \mathcal{G}'}({\bf x})$, and we are given the set of steps where inversions occur in the dendrogram of ${\bf x}$ below the $(n-K)$th merge, 
$$\mathcal{M}({\bf x}) = \left \{ 1 \leq t < n-K : d \left ( \mathcal{W}_1^{(t)}({\bf x}), \mathcal{W}_2^{(t)}({\bf x}) ; {\bf x}\right ) > d \left ( \mathcal{W}_1^{(t+1)}({\bf x}), \mathcal{W}_2^{(t+1)}({\bf x}) ; {\bf x}\right ) \right \}.$$

Observe that for $h_{\mathcal{G}, \mathcal{G}'}({\bf x})$ in \eqref{eq:hab}, 
\begin{align*} 
h_{\mathcal{G}, \mathcal{G}'}({\bf x}) &= \underset{l_{\mathcal{G}, \mathcal{G}'}({\bf x}) \leq t \leq u_{\mathcal{G}, \mathcal{G}'}({\bf x})}{\max} d\left ( \mathcal{W}_1^{(t)}({\bf x}),\mathcal{W}_2^{(t)}({\bf x}); {\bf x} \right ) = \underset{t \in \mathcal{M}_{\mathcal{G}, \mathcal{G}'}({\bf x}) \cup \{u_{\mathcal{G}, \mathcal{G}'}({\bf x}) \}}{\max} d \left (\mathcal{W}_1^{(t)}({\bf x}),   \mathcal{W}_2^{(t)}({\bf x}); {\bf x} \right ),
\end{align*} 
where 
$\mathcal{M}_{\mathcal{G}, \mathcal{G}'}({\bf x}) = \left \{ t: l_{\mathcal{G}, \mathcal{G}'}({\bf x}) \leq t < u_{\mathcal{G}, \mathcal{G}'}({\bf x}),  d \left ( \mathcal{W}_1^{(t)}({\bf x}), \mathcal{W}_2^{(t)}({\bf x}) ; {\bf x}\right ) > d \left ( \mathcal{W}_1^{(t+1)}({\bf x}), \mathcal{G}_2^{(l+1)}({\bf x}) ; {\bf x}\right ) \right \}.$
Therefore, we can compute $h_{\mathcal{G}, \mathcal{G}'} ({\bf x})$ for each cluster pair $\{\mathcal{G}, \mathcal{G}'\} \in \mathcal{L}({\bf x})$ as follows: 
\begin{enumerate}
\item Compute $\mathcal{M}_{\mathcal{G}, \mathcal{G}'}({\bf x})$ by checking every step in $\mathcal{M}({\bf x})$ to see if it is in $\Big [l_{\mathcal{G}, \mathcal{G}'}({\bf x}),  u_{\mathcal{G}, \mathcal{G}'}({\bf x}) \Big )$. 
\item Compute $h_{\mathcal{G}, \mathcal{G}'}({\bf x})$ by taking the max over $|\mathcal{M}_{\mathcal{G}, \mathcal{G}'}({\bf x})| + 1$ merge heights. 
\end{enumerate} 
Step 1 requires $2|\mathcal{M}({\bf x})|$ operations, and Step 2 requires $|\mathcal{M}_{\mathcal{G}, \mathcal{G}'}({\bf x})| + 1$ operations. Since  $|\mathcal{M}_{\mathcal{G}, \mathcal{G}'}({\bf x})| <  |\mathcal{M}({\bf x})|$, it follows that we can compute $h_{\mathcal{G}, \mathcal{G}'} ({\bf x})$ in $\mathcal{O}(|\mathcal{M}({\bf x})|+ 1)$ time.

\subsection{Proof of Proposition \ref{prop:single}} 
\label{sec:proof-single}
Recall that 
\begin{align} 
\mathcal{\hat S}_{i, i'} = \left \{ \phi \geq 0: d(\{i\}, \{i'\}; {\bf x}'(\phi)) > \underset{\{\mathcal{G}, \mathcal{G}'\} \in \mathcal{L}({\bf x}):  i \in \mathcal{G}, i' \in \mathcal{G}'}{\max}~ h_{\mathcal{G}, \mathcal{G}'}({\bf x}) \right \}. \label{eq:defSpair}
\end{align} 

We will consider two cases: either 
\begin{align} 
 i, i' \in \hc_1 \text{ or } i, i' \in \hc_2 \text{ or } i, i' \not \in \hc_1 \cup \hc_2, \label{eq:pair-cond}
\end{align} 
or \eqref{eq:pair-cond} doesn't hold. 

\subsubsection*{Case 1: \eqref{eq:pair-cond} holds}
We first state and prove a preliminary result.  
\begin{lemma} 
\label{lem:single-preserve}
Suppose that $\mathcal{C} = \mathcal{C}^{(n-K+1)}$, i.e. we perform hierarchical clustering to obtain $K$ clusters. Let $\{\mathcal{G}, \mathcal{G}'\} \in \mathcal{L}({\bf x})$, $i \in \mathcal{G}$, $i' \in \mathcal{G}'$. Suppose that \eqref{eq:pair-cond} holds. Then, for all $\phi \geq 0$, 
\begin{align} 
d(\mathcal{G}, \mathcal{G}'; {\bf x}'(\phi)) = d(\mathcal{G}, \mathcal{G}'; {\bf x}). \label{eq:single-preserve}
\end{align} 
\end{lemma} 
\begin{proof} 
Since  $\{\mathcal{G}, \mathcal{G}'\} \in \mathcal{L}({\bf x})$, it follows that $\mathcal{G}, \mathcal{G}' \in \bigcup \limits_{t=1}^{n-K} \mathcal{C}^{(t)}({\bf x})$. Furthermore, clusters that merge cannot unmerge, and $\hc_1, \hc_2 \in \mathcal{C}^{(n-K+1)}({\bf x})$. This means that all observations in the set $\mathcal{G}$ must belong to the same cluster in  $\mathcal{C}^{(n-K+1)}({\bf x})$, and all observations  in the set $\mathcal{G'}$ must belong to the same cluster in $\mathcal{C}^{(n-K+1)}({\bf x})$. Therefore, \eqref{eq:pair-cond} implies that either
\begin{align*} 
\mathcal{G}, \mathcal{G}' \subseteq \hc_1 \text{ or } \mathcal{G}, \mathcal{G}' \subseteq \hc_2 \text{ or }  \mathcal{G}, \mathcal{G}' \subseteq \{1, 2, \ldots, n\} \backslash [\hc_1 \cup \hc_2].
\end{align*} 
It follows from Lemma \ref{lem:dissimilar} that \eqref{eq:single-preserve} holds. 
\end{proof} 

Suppose that $i, i'$ satisfy \eqref{eq:pair-cond}. Recall from Section \ref{sec:key} that 
\begin{align} 
d(\mathcal{G}, \mathcal{G}'; {\bf x}) > h_{\mathcal{G}, \mathcal{G}'}({\bf x}) \quad \forall ~ \{\mathcal{G}, \mathcal{G}'\} \in \mathcal{L}({\bf x}). \label{eq:single-ge} 
\end{align}
It follows from \eqref{eq:single-ge} and Lemma \ref{lem:single-preserve} that for all $\phi \geq 0$, 
\begin{align} 
d(\mathcal{G}, \mathcal{G}'; {\bf x}'(\phi)) = d(\mathcal{G}, \mathcal{G}'; {\bf x}) > h_{\mathcal{G}, \mathcal{G}'}({\bf x}) \quad \forall ~\{\mathcal{G}, \mathcal{G}'\} \in \mathcal{L}({\bf x}) \text{ s.t. } i \in \mathcal{G}, i' \in \mathcal{G}'. \label{eq:single-ge2} 
\end{align} 
Recall that single linkage defines $d(\mathcal{G}, \mathcal{G}'; {\bf x}'(\phi)) = \underset{ i\in \mathcal{G}, i' \in \mathcal{G}'}{\min} d(\{i\}, \{i'\}; {\bf x}'(\phi))$. Thus, it follows from \eqref{eq:single-ge2} that for all $\phi \geq 0$, 
\begin{align} 
d(\{i\}, \{i'\}; {\bf x}'(\phi)) > h_{\mathcal{G}, \mathcal{G}'}({\bf x}) \quad \forall ~\{\mathcal{G}, \mathcal{G}'\} \in \mathcal{L}({\bf x}) \text{ s.t. } i \in \mathcal{G}, i' \in \mathcal{G}'.  \label{eq:single-pf}
\end{align} 

Applying \eqref{eq:single-pf} to the definition of $\mathcal{\hat S}_{i, i'}$ in \eqref{eq:defSpair} yields $\mathcal{\hat S}_{i, i'} = [0, \infty)$. 

\subsubsection*{Case 2: \eqref{eq:pair-cond} does not hold}
Suppose that $i, i'$ do not satisfy \eqref{eq:pair-cond}. Then, $i$ and $i'$ are assigned to different clusters in $\mathcal{C}({\bf x}) = \mathcal{C}^{(n-K+1)}({\bf x})$. Since $\mathcal{C}^{(n-K+1)}({\bf x})$ is created by merging a pair of clusters in $\mathcal{C}^{(n-K)}({\bf x})$, it follows that $i$ and $i'$ are assigned to different clusters $\mathcal{G}$ and $\mathcal{G}'$ in $\mathcal{C}^{(n-K)}({\bf x})$ such that $\{\mathcal{G}, \mathcal{G}'\}$ is not the winning pair. That is, there exists $\mathcal{G}, \mathcal{G}' \in \mathcal{C}^{(n-K)}({\bf x})$ such that $i \in \mathcal{G}$, and $i' \in \mathcal{G}'$, $\mathcal{G} \neq \mathcal{G}'$, and $\{\mathcal{G}, \mathcal{G}' \} \neq  \left \{\mathcal{W}_1^{(n-K)}({\bf x}), \mathcal{W}_2^{(n-K)}({\bf x} ) \right \}$. This means that $\{\mathcal{G}, \mathcal{G}'\} \in \mathcal{L}({\bf x})$, with $u_{\mathcal{G}, \mathcal{G}'}({\bf x}) = n-K$. Furthermore, single linkage cannot produce inversions \citep{murtagh2012algorithms}, i.e. $d\left (\mathcal{W}_1^{(t)}({\bf x}), \mathcal{W}_2^{(t)}({\bf x}); {\bf x} \right ) < d\left (\mathcal{W}_1^{(t+1)}({\bf x}), \mathcal{W}_2^{(t+1)}({\bf x}); {\bf x} \right )$ for all $t < n-1$. It follows that
\begin{align} 
\underset{\{\mathcal{G}, \mathcal{G}'\} \in \mathcal{L}({\bf x}): i \in \mathcal{G}, i' \in \mathcal{G}'}{\max}~ h_{\mathcal{G}, \mathcal{G}'}({\bf x})  &= \underset{\mathcal{G}, \mathcal{G}' \in \mathcal{L}({\bf x}): i \in \mathcal{G}, i' \in \mathcal{G}'}{\max} ~ \underset{l_{\mathcal{G}, \mathcal{G}'}({\bf x}) \leq t \leq u_{\mathcal{G}, \mathcal{G}'}({\bf x})}{\max} ~ d\left (\mathcal{W}_1^{(t)}({\bf x}), \mathcal{W}_2^{(t)}({\bf x}); {\bf x} \right ) \nonumber \\ 
&= d\left (\mathcal{W}^{(n-K)}_1({\bf x}), \mathcal{W}^{(n-K)}_2({\bf x}); {\bf x} \right ). \label{eq:single-proof} 
\end{align} 
 Applying \eqref{eq:single-proof} to the definition of $\mathcal{\hat S}_{i, i'}$ in \eqref{eq:defSpair} yields 
$$\mathcal{\hat S}_{i, i'} = \left \{ \phi \geq 0: d(\{i\}, \{i'\}; {\bf x}'(\phi)) > d\left (\mathcal{W}^{(n-K)}_1({\bf x}), \mathcal{W}^{(n-K)}_2({\bf x}); {\bf x} \right ) \right \}. $$

\subsection{Proof of Theorem \ref{thm:plug-in-conservative}} 
\label{sec:proof-plug-in}

The proof of Theorem \ref{thm:plug-in-conservative} is adapted from Appendices A.10 and A.11 in \citet{tibshirani2018uniform}. 
We first state and prove an intermediate result. 
\begin{lemma} 
\label{lem:monotone}
Recall that $F(t; c, \mathcal{S})$ denotes the CDF of a $c \cdot \chi_q$ random variable truncated to the set $\mathcal{S}$.  If $0 < c_1 < c_2$, then
$$1 - F(t; c_1, \mathcal{S}) < 1 - F(t; c_2, \mathcal{S}), \quad \text{for all } t \in \mathcal{S}.$$
\end{lemma} 
\begin{proof} 
Let $f(t; c, \mathcal{S})$ denote the probability density function (pdf) of a $c \cdot \chi_q$ random variable truncated to the set $\mathcal{S}$. Observe that 
\begin{align*} 
f(t; c, \mathcal{S}) &= \frac{ \frac{1}{2^{\frac{q}{2} - 1} \Gamma(q/2)} c^{-q} t^{q-1} e^{-\frac{t^2}{2c^2}} \mathds{1} \{t \in \mathcal{S}\}}{ \int \frac{1}{2^{\frac{q}{2} - 1} \Gamma(q/2)} c^{-q} u^{q-1} e^{-\frac{u^2}{2c^2}} \mathds{1} \{u \in \mathcal{S}\} ~ du } \\
&= \frac{t^{q-1} e^{-\frac{t^2}{2c^2}} \mathds{1} \{t \in \mathcal{S}\}}{\int u^{q-1} e^{-\frac{u^2}{2c^2}} \mathds{1} \{u \in \mathcal{S}\} ~ du  }  \\ 
&= \left ( \frac{t^{q-1}  \mathds{1} \{t \in \mathcal{S}\}}{\int u^{q-1} e^{-\frac{u^2}{2c^2}} \mathds{1} \{u \in \mathcal{S}\} ~ du}  \right ) \left ( \exp \left \{ -\frac{t^2}{2c^2} \right \} \right ). 
\end{align*} 
Thus, $\{f(t; c, \mathcal{S})\}_{c > 0}$ is an exponential family with natural parameter $\frac{1}{c^2}$. Therefore, it has a monotone non-increasing likelihood ratio in its sufficient statistic, $-t^2/2$. This means that for any fixed $0 < c_1 < c_2$, 
\begin{align}
\frac{f(u_2; c_1, \mathcal{S})}{f(u_2; c_2, \mathcal{S}) } < \frac{f(u_1; c_1, \mathcal{S})}{f(u_1; c_2, \mathcal{S}) }, \quad \text{for all } u_1, u_2 \in \mathcal{S}, ~ u_1 < u_2. \label{eq:mlr} 
\end{align} 
Rearranging terms in \eqref{eq:mlr} yields: 
\begin{align} 
f(u_2; c_1, \mathcal{S})f(u_1; c_2, \mathcal{S}) < f(u_1; c_1, \mathcal{S})f(u_2; c_2, \mathcal{S}), \quad \text{for all } u_1, u_2 \in \mathcal{S}, ~ u_1 < u_2.\label{eq:mlr2} 
\end{align} 
Let $t, u_2 \in \mathcal{S}$ with $t < u_2$. Then, 
\begin{align*} 
f(u_2; c_1, \mathcal{S})F(t; c_2, \mathcal{S}) &= \int_{-\infty}^t f(u_2; c_1, \mathcal{S}) f(u_1; c_2, \mathcal{S}) du_1 \nonumber \\ 
&< \int_{-\infty}^t f(u_1; c_1, \mathcal{S})f(u_2; c_2, \mathcal{S}) du_1 =  F(t; c_1, \mathcal{S})f(u_2; c_2, \mathcal{S}), 
\end{align*}  
where the inequality follows from \eqref{eq:mlr2}. That is, we have shown that
\begin{align} 
f(u_2; c_1, \mathcal{S})F(u_1; c_2, \mathcal{S}) < F(u_1; c_1, \mathcal{S})f(u_2; c_2, \mathcal{S}),  \quad \text{for all } u_1, u_2 \in \mathcal{S}, ~ u_1 < u_2. \label{eq:mlr3} 
\end{align} 
Let $t \in \mathcal{S}$. Then, 
\begin{align} 
(1 - F(t; c_1, \mathcal{S}))F(t; c_2, \mathcal{S}) &= \int_t^\infty f(u_2; c_1, \mathcal{S})F(t; c_2, \mathcal{S}) du_2  \label{eq:mlr4} \\ 
&< \int_t^\infty F(t; c_1, \mathcal{S})f(u_2; c_2, \mathcal{S}) du_2 = F(t; c_1, \mathcal{S})(1 -  F(t; c_2, \mathcal{S})), \nonumber
\end{align}
where the inequality follows from \eqref{eq:mlr3}  and the fact that $f(t; c, \mathcal{S}) = 0$ for all $t \not \in \mathcal{S}$. Rearranging terms in \eqref{eq:mlr4} yields  
$$1 - F(t; c_1, \mathcal{S}) < 1 - F(t; c_2, \mathcal{S}).$$
\end{proof} 
We will now  prove Theorem \ref{thm:plug-in-conservative}. In the following, we will use $A_n$ to denote the event $ \left \{ \hc_1^{(n)}, \hc_2^{(n)} \in \mathcal{C} ( {\bf X}^{(n)} )  \right \} $, $\hat p_n$ to denote $\hat p\left ({\bf X}^{(n)}; \left \{\hc_1^{(n)}, \hc_2^{(n)} \right \} \right )$, and $p_n$ to denote $p\left ({\bf X}^{(n)}; \left \{\hc_1^{(n)}, \hc_2^{(n)} \right \} \right )$. 
Observe that 
\begin{align} 
& \mathbb{P}_{H_0^{ \left \{\hc_1^{(n)}, \hc_2^{(n)} \right \}}} \left ( \hat p_n \leq \alpha ~\Big | ~ A_n \right ) \nonumber \\ 
&= \mathbb{P}_{H_0^{ \left \{\hc_1^{(n)}, \hc_2^{(n)} \right \}}} \left ( \hat p_n \leq \alpha, ~ \hat \sigma({\bf X}^{(n)}) \geq \sigma ~\Big | ~ A_n \right ) + \mathbb{P}_{H_0^{ \left \{\hc_1^{(n)}, \hc_2^{(n)} \right \}}} \left ( \hat p_n \leq \alpha, ~ \hat \sigma({\bf X}^{(n)}) < \sigma ~\Big | ~ A_n \right ) \nonumber \\ 
&\leq\mathbb{P}_{H_0^{ \left \{\hc_1^{(n)}, \hc_2^{(n)} \right \}}} \left ( \hat p_n \leq \alpha, ~ \hat \sigma({\bf X}^{(n)}) \geq \sigma ~\Big | ~ A_n \right ) + \mathbb{P}_{H_0^{ \left \{\hc_1^{(n)}, \hc_2^{(n)} \right \}}} \left ( \hat \sigma({\bf X}^{(n)}) < \sigma ~\Big | ~ A_n \right ).\label{eq:finite1}
\end{align} 
Recall from \eqref{eq:pval-phi} that 
$$ p\left ({\bf x}^{(n)}; \left \{\hc_1^{(n)}, \hc_2^{(n)} \right \} \right )  = 1 - \mathbb{F} \left (\left \|\bar{x}_{\hc_1^{(n)}} - \bar{x}_{\hc_2^{(n)}} \right \|_2; \sigma \sqrt{\frac{1}{|\hc_1^{(n)}|} +\frac{1}{|\hc_2^{(n)}|}} , \mathcal{ S}\left ({\bf x};  \left \{\hc_1^{(n)}, \hc_2^{(n)} \right \} \right ) \right ), $$
and recall from \eqref{eq:est-pval} that
$$\hat p\left ({\bf x}^{(n)}; \left \{\hc_1^{(n)}, \hc_2^{(n)} \right \} \right )  = 1 - \mathbb{F} \left (\left \|\bar{x}_{\hc_1^{(n)}} - \bar{x}_{\hc_2^{(n)}} \right \|_2; \hat \sigma ( {\bf x})  \sqrt{\frac{1}{|\hc_1^{(n)}|} +\frac{1}{|\hc_2^{(n)}|}} , \mathcal{ S}\left ({\bf x};  \left \{\hc_1^{(n)}, \hc_2^{(n)} \right \} \right ) \right ).$$
It follows from Lemma \ref{lem:monotone} that for any ${\bf x}$ such that $\hat \sigma({\bf x}) \geq \sigma$, we have that $$\hat p\left ({\bf x}; \left \{\hc_1^{(n)}, \hc_2^{(n)} \right \} \right ) \geq p \left ({\bf x}; \left \{\hc_1^{(n)}, \hc_2^{(n)} \right \} \right ).$$ Thus, for all $n = 1, 2, \ldots$ and for all $0 \leq \alpha \leq 1$, 
\begin{align} \mathbb{P}_{H_0^{ \left \{\hc_1^{(n)}, \hc_2^{(n)} \right \}}} \left (\hat p_n \leq \alpha,  \hat \sigma ({\bf X}^{(n)}) \geq \sigma  ~\Big | ~ A_n \right ) &\leq \mathbb{P}_{H_0^{ \left \{\hc_1^{(n)}, \hc_2^{(n)} \right \}}} \left ( p_n \leq \alpha,  \hat \sigma ({\bf X}^{(n)}) \geq \sigma  ~\Big | ~ A_n \right )  \nonumber \\ 
&\leq \mathbb{P}_{H_0^{ \left \{\hc_1^{(n)}, \hc_2^{(n)} \right \}}} \Big (p_n \leq \alpha ~ \Big | ~ A_n \Big ) \nonumber \\ 
&= \alpha, \label{eq:ineq}
\end{align} 
where the equality follows from \eqref{eq:selective-ours}. Applying \eqref{eq:ineq} to \eqref{eq:finite1} yields 
\begin{align} 
&\mathbb{P}_{H_0^{ \left \{\hc_1^{(n)}, \hc_2^{(n)} \right \}}} \left ( \hat p_n \leq \alpha ~\Big | ~ A_n \right )  \leq \alpha + \mathbb{P}_{H_0^{ \left \{\hc_1^{(n)}, \hc_2^{(n)} \right \}}} \left ( \hat \sigma({\bf X}^{(n)}) < \sigma ~\Big | ~ A_n \right ).\label{eq:finite2}
\end{align} 
By our assumption in Theorem \ref{thm:plug-in-conservative}, $\underset{n \rightarrow \infty}{\lim} ~ \mathbb{P}_{H_0^{ \left \{\hc_1^{(n)}, \hc_2^{(n)} \right \}}} \left ( \hat \sigma({\bf X}^{(n)}) < \sigma ~\Big | ~ A_n \right ) = 0$. Applying this fact to \eqref{eq:finite2} yields 
$$ \underset{n \rightarrow \infty}{\lim} \mathbb{P}_{H_0^{ \left \{\hc_1^{(n)}, \hc_2^{(n)} \right \}}} \left ( \hat p_n \leq \alpha ~\Big | ~ A_n \right ) \leq \alpha.$$
This completes the proof of Theorem \ref{thm:plug-in-conservative}.

\subsection{Proof of Propositions \ref{prop:random1} and \ref{prop:random2}}
\label{sec:proof-random1} 

Recall that $H_0^{\{\mathcal{G}_1, \mathcal{G}_2\}}$ denotes the null hypothesis that $\bar{\mu}_{\mathcal{G}_1} = \bar{\mu}_{\mathcal{G}_2}$. Further recall that $\mathcal{H}_0$ denotes the set of hypotheses of the form $H_0^{\{\mathcal{G}_1, \mathcal{G}_2\}}$ that are true. We start by rewriting \eqref{eq:selective-ours} as follows. 
\begin{align} \text{If } H_0^{\{\mathcal{G}_1, \mathcal{G}_2\}} \in \mathcal{H}_0, \text{ then } \mathbb{P}\left ( p({\bf X}; \{\mathcal{G}_1, \mathcal{G}_2\}) \leq \alpha ~\Big |  ~\mathcal{G}_1, \mathcal{G}_2 \in  \mathcal{C}({\bf X})  \right ) = \alpha. \label{eq:selective-appendix} 
\end{align} 

We will first prove Proposition \ref{prop:random1}. Let $K = 2$. Then, 
\begin{align} 
&\mathbb{P}\left (p({\bf X}; \mathcal{C}({\bf X})) \leq \alpha, H_0^{\mathcal{C}({\bf X})} \in \mathcal{H}_0  \right ) \nonumber \\ 
&= \sum \limits_{H_0^{\{\mathcal{G}_1, \mathcal{G}_2\}} \in \mathcal{H}_0} \mathbb{P}\left ( p({\bf X}; \{\mathcal{G}_1, \mathcal{G}_2\}) \leq \alpha, \mathcal{C}({\bf X}) = \{ \mathcal{G}_1, \mathcal{G}_2\} \right ) \nonumber\\ 
&= \sum \limits_{H_0^{\{\mathcal{G}_1, \mathcal{G}_2\}} \in \mathcal{H}_0} \mathbb{P}\left (p({\bf X}; \{\mathcal{G}_1, \mathcal{G}_2\}) \leq \alpha \mid \mathcal{C}({\bf X}) = \{ \mathcal{G}_1, \mathcal{G}_2\} \right ) \mathbb{P}(\mathcal{C}({\bf X}) = \{ \mathcal{G}_1, \mathcal{G}_2\})\nonumber \\ 
&=  \sum \limits_{H_0^{\{\mathcal{G}_1, \mathcal{G}_2\}} \in \mathcal{H}_0}  \alpha ~ \mathbb{P}(\mathcal{C}({\bf X}) = \{ \mathcal{G}_1, \mathcal{G}_2\}) \label{eq:apply-selective1} \\ 
&= \alpha ~ \mathbb{P}(H_0^{\mathcal{C}({\bf X})} \in \mathcal{H}_0), \nonumber
\end{align}
where \eqref{eq:apply-selective1} follows from \eqref{eq:selective-appendix} and the fact that $K = 2$. Therefore, 
$$\mathbb{P}\left (p({\bf X}; \mathcal{C}({\bf X})) \leq \alpha \mid  H_0^{\mathcal{C}({\bf X})} \in \mathcal{H}_0  \right ) = \alpha. $$
This is \eqref{eq:random1}. 

We will now prove Proposition \ref{prop:random2}. Let $K > 2$, and let $\{\mathcal{G}_1({\bf X}), \mathcal{G}_2({\bf X})\}$ and ${\bf X}$ be conditionally independent given $\mathcal{C}({\bf X})$. Then, 
\begin{align} 
&\mathbb{P}\left (p({\bf X};  \{\mathcal{G}_1({\bf X}), \mathcal{G}_2({\bf X}) \}) \leq \alpha, ~  H_0^{\left \{ \mathcal{G}_1({\bf X}), \mathcal{G}_2({\bf X}) \right \}}  \in \mathcal{H}_0 \right )  \label{eq:start}\\ 
&= \sum \limits_{H_0^{\{\mathcal{G}_1, \mathcal{G}_2\}} \in \mathcal{H}_0} \mathbb{P} \left (  p({\bf X}; \{\mathcal{G}_1, \mathcal{G}_2 \}) \leq \alpha, \{\mathcal{G}_1({\bf X}), \mathcal{G}_2({\bf X})\} = \{ \mathcal{G}_1, \mathcal{G}_2\} \right ) \nonumber  \\ 
&= \sum \limits_{H_0^{\{\mathcal{G}_1, \mathcal{G}_2\}} \in \mathcal{H}_0} \mathbb{P}(\{ p({\bf X}; \{\mathcal{G}_1, \mathcal{G}_2 \}) \leq \alpha \} \cap \{ \mathcal{G}_1, \mathcal{G}_2 \in \mathcal{C}({\bf X}) \} \cap \{ \{\mathcal{G}_1({\bf X}), \mathcal{G}_2({\bf X})\} = \{\mathcal{G}_1, \mathcal{G}_2\}\}  )  \nonumber
\end{align} 
where the last equality follows from the fact that $\{\mathcal{G}_1({\bf X}), \mathcal{G}_2({\bf X})\} = \{\mathcal{G}_1, \mathcal{G}_2\}$ implies that $\mathcal{G}_1, \mathcal{G}_2 \in \mathcal{C}({\bf X})$. Since we assumed that $\{\mathcal{G}_1({\bf X}), \mathcal{G}_2({\bf X})\}$ and ${\bf X}$ are conditionally independent given $\mathcal{C}({\bf X})$, we have for any $H_0^{\{\mathcal{G}_1, \mathcal{G}_2\}} \in \mathcal{H}_0$ that 
\begin{align*}
& \mathbb{P} \left ( p({\bf X}; \{\mathcal{G}_1, \mathcal{G}_2 \}) \leq \alpha,  \{\mathcal{G}_1({\bf X}), \mathcal{G}_2({\bf X})\} = \{\mathcal{G}_1, \mathcal{G}_2\} \mid \mathcal{G}_1, \mathcal{G}_2 \in \mathcal{C}({\bf X}) \right )  \\ 
 &= \mathbb{P} \left ( p({\bf X}; \{\mathcal{G}_1, \mathcal{G}_2\}) \leq \alpha \mid \mathcal{G}_1, \mathcal{G}_2 \in \mathcal{C}({\bf X}) \right ) \mathbb{P} \left (\{\mathcal{G}_1({\bf X}), \mathcal{G}_2({\bf X})\} = \{\mathcal{G}_1, \mathcal{G}_2\} \mid \mathcal{G}_1, \mathcal{G}_2 \in \mathcal{C}({\bf X})  \right )  \\ 
 &= \alpha ~ \mathbb{P} \left (\{\mathcal{G}_1({\bf X}), \mathcal{G}_2({\bf X})\} = \{\mathcal{G}_1, \mathcal{G}_2\} \mid \mathcal{G}_1, \mathcal{G}_2 \in \mathcal{C}({\bf X})  \right ), 
\end{align*}
where the last equality follows from \eqref{eq:selective-appendix}. It follows that for any $H_0^{\{\mathcal{G}_1, \mathcal{G}_2\}} \in \mathcal{H}_0$, we have that 
\begin{align} 
&\mathbb{P}(\{ p({\bf X}; \{\mathcal{G}_1, \mathcal{G}_2 \}) \leq \alpha \} \cap \{ \mathcal{G}_1, \mathcal{G}_2 \in \mathcal{C}({\bf X}) \} \cap \{ \{\mathcal{G}_1({\bf X}), \mathcal{G}_2({\bf X})\} = \{\mathcal{G}_1, \mathcal{G}_2\}\}  )  \nonumber \\ 
&= \alpha ~ \mathbb{P} \left (\{\mathcal{G}_1({\bf X}), \mathcal{G}_2({\bf X})\} = \{\mathcal{G}_1, \mathcal{G}_2\} \mid \mathcal{G}_1, \mathcal{G}_2 \in \mathcal{C}({\bf X})  \right ) \mathbb{P} \left ( \mathcal{G}_1, \mathcal{G}_2 \in \mathcal{C}({\bf X})  \right ) \nonumber \\ 
&= \alpha ~\mathbb{P} \left ( \{\mathcal{G}_1({\bf X}), \mathcal{G}_2({\bf X})\} = \{\mathcal{G}_1, \mathcal{G}_2\},  \mathcal{G}_1, \mathcal{G}_2 \in \mathcal{C}({\bf X}) \right ). \label{eq:apply-selective2} 
\end{align}  
Applying \eqref{eq:apply-selective2} to \eqref{eq:start} yields 
\begin{align*} 
&\mathbb{P}\left (p({\bf X};  \{\mathcal{G}_1({\bf X}), \mathcal{G}_2({\bf X}) \}) \leq \alpha, ~  H_0^{\left \{ \mathcal{G}_1({\bf X}), \mathcal{G}_2({\bf X}) \right \}}  \in \mathcal{H}_0 \right )  \\ 
&= \sum \limits_{H_0^{\{\mathcal{G}_1, \mathcal{G}_2\}} \in \mathcal{H}_0}  \alpha ~ \mathbb{P} \left ( \{\mathcal{G}_1({\bf X}), \mathcal{G}_2({\bf X})\} = \{\mathcal{G}_1, \mathcal{G}_2\},  \mathcal{G}_1, \mathcal{G}_2 \in \mathcal{C}({\bf X}) \right ) \\ 
&=  \sum \limits_{H_0^{\{\mathcal{G}_1, \mathcal{G}_2\}} \in \mathcal{H}_0}  \alpha ~ \mathbb{P} \left ( \{\mathcal{G}_1({\bf X}), \mathcal{G}_2({\bf X})\} = \{\mathcal{G}_1, \mathcal{G}_2\} \right ) \\ 
&= \alpha ~ \mathbb{P} \left ( H_0^{\left \{ \mathcal{G}_1({\bf X}), \mathcal{G}_2({\bf X}) \right \}}  \in \mathcal{H}_0 \right ), 
\end{align*} 
where the second-to-last equality follows from the fact that $\{\mathcal{G}_1({\bf X}), \mathcal{G}_2({\bf X})\} = \{\mathcal{G}_1, \mathcal{G}_2\}$ implies that $\mathcal{G}_1, \mathcal{G}_2 \in \mathcal{C}({\bf X})$. 
Therefore, 
$$\mathbb{P}\left (p({\bf X};  \{\mathcal{G}_1({\bf X}), \mathcal{G}_2({\bf X}) \}) \leq \alpha \mid  H_0^{\left \{ \mathcal{G}_1({\bf X}), \mathcal{G}_2({\bf X}) \right \}}  \in \mathcal{H}_0 \right ) = \alpha. $$
This is \eqref{eq:random2}.

\section{Algorithm for computing $\mathcal{\hat S}$ in the case of squared Euclidean distance and linkages that satisfy \eqref{eq:lance1}} 
\label{sec:lance-algo} 
We begin by defining quantities used in the algorithm. For all $\mathcal{G} \in \bigcup \limits_{t=1}^{n-K} \mathcal{C}^{(t)}({\bf x})$, let
\begin{align} 
l_{\mathcal{G}}({\bf x}) \equiv \min \{1 \leq t \leq n-K: \mathcal{G} \in \mathcal{C}^{(t)}({\bf x})\}, ~ u_{\mathcal{G}}({\bf x}) \equiv \max \{1 \leq t \leq n-K: \mathcal{G} \in \mathcal{C}^{(t)}({\bf x})\}, \label{eq:lifetimes1} 
\end{align}
so that for any $\{ \mathcal{G}, \mathcal{G}' \} \in \mathcal{L}({\bf x})$, 
\begin{align} 
l_{\mathcal{G}, \mathcal{G}'}({\bf x}) &= \max \{ l_{\mathcal{G}}({\bf x}), l_{\mathcal{G}'}({\bf x})\}, \label{eq:lifetimes2}\\ 
 u_{\mathcal{G}, \mathcal{G}'}({\bf x}) &= \begin{cases} \min \{ u_{\mathcal{G}}({\bf x}), u_{\mathcal{G}'}({\bf x})\} - 1, & \text{if } \{\mathcal{G}, \mathcal{G}'\} \in \left \{ \left \{ \mathcal{W}_1^{(t)}({\bf x}), \mathcal{W}_2^{(t)}({\bf x}) \right \} : 1 \leq t \leq n-K \right \}, \\  \min \{ u_{\mathcal{G}}({\bf x}), u_{\mathcal{G}'}({\bf x})\}, & \text{ otherwise.} \end{cases}  \nonumber
\end{align} 
Recall from Section \ref{sec:lance} that 
\begin{align} 
h_{\mathcal{G}, \mathcal{G}'}({\bf x}) &= \underset{t \in \mathcal{M}_{\mathcal{G}, \mathcal{G}'}({\bf x}) \cup \{u_{\mathcal{G}, \mathcal{G}'}({\bf x}) \}}{\max} d \left (\mathcal{W}_1^{(t)}({\bf x}),   \mathcal{W}_2^{(t)}({\bf x}); {\bf x}\right ), \text{ where } \label{eq:h}\\ 
\mathcal{M}_{\mathcal{G}, \mathcal{G}'}({\bf x}) &= \left \{t: t \in \mathcal{M}({\bf x}), l_{\mathcal{G}, \mathcal{G}'}({\bf x}) \leq t < u_{\mathcal{G}, \mathcal{G}'}({\bf x}) \right \}. \label{eq:m}
\end{align} 

Let $\mathcal{G}_{new}^{(t)}({\bf x}) = \mathcal{W}_1^{(t-1)}({\bf x}) \cup \mathcal{W}_2^{(t-1)}({\bf x})$ for all $t = 2, \ldots, n-K$. Define
\begin{align} 
\mathcal{L}_1({\bf x}) &= \left \{ \{ \mathcal{G}, \mathcal{G}' \}: \mathcal{G}, \mathcal{G}' \in \mathcal{C}^{(1)}({\bf x}), \mathcal{G} \neq \mathcal{G}', \{\mathcal{G}, \mathcal{G}' \} \neq \{\mathcal{W}^{(1)}_1({\bf x}), \mathcal{W}^{(1)}_2({\bf x}) \} \right \},\label{eq:losers1}
\end{align} 
and for $t = 2, \ldots, n-K$, define 
\begin{align} 
\mathcal{L}_t({\bf x}) &= \left \{ \left \{ \mathcal{G}_{new}^{(t)}({\bf x}) , \mathcal{G}' \right \}: \mathcal{G}' \in \mathcal{C}^{(t)}({\bf x}) \backslash \mathcal{G}_{new}^{(t)}({\bf x}) ,  \left \{ \mathcal{G}_{new}^{(t)}({\bf x}) , \mathcal{G}' \right \} \neq \left \{\mathcal{W}^{(t)}_1({\bf x}), \mathcal{W}^{(t)}_2({\bf x}) \right \} \right \},\label{eq:losersl}
\end{align} 
so that for $\mathcal{L}({\bf x})$ defined in \eqref{eq:losers},  $\mathcal{L}({\bf x}) = \bigcup \limits_{t=1}^{n-K} \mathcal{L}_t({\bf x}).$ Thus, it follows from Theorem \ref{thm:hierS} that 
\begin{align} 
\mathcal{\hat S} = \bigcap \limits_{t=1}^{n-K} \bigcap \limits_{\{ \mathcal{G}, \mathcal{G}' \} \in \mathcal{L}_t ({\bf x})}\mathcal{\hat S}_{\mathcal{G}, \mathcal{G}'} ,  \label{eq:Sintersect} 
\end{align} 
where for all $\{\mathcal{G}, \mathcal{G}'\} \in \mathcal{L}({\bf x})$, 
\begin{align} 
\mathcal{\hat S}_{\mathcal{G}, \mathcal{G}'} =  \{\phi \geq 0: d(\mathcal{G}, \mathcal{G}'; {\bf x}'(\phi)) > h_{\mathcal{G}, \mathcal{G}'}({\bf x}) \}. \label{eq:Ssub} 
\end{align} 

The following algorithm computes the set $\mathcal{\hat S}$ (defined in \eqref{eq:defS}) in the case of squared Euclidean distance and linkages that satisfy \eqref{eq:lance1} in $\mathcal{O}\Big ((|\mathcal{M}({\bf x})| + \log(n))n^2 \Big )$, where $|\mathcal{M}({\bf x})|$ is the number of inversions in the dendrogram of ${\bf x}$ below the $(n-K)$th merge:
\begin{enumerate} 
\item Compute $\mathcal{M}({\bf x}) = \left \{ 1 \leq t < n-K: d \left ( \mathcal{W}_1^{(t)}({\bf x}), \mathcal{W}_2^{(t)}({\bf x}) ; {\bf x}\right ) > d \left ( \mathcal{W}_1^{(t+1)}({\bf x}), \mathcal{W}_2^{(t+1)}({\bf x}) ; {\bf x}\right )  \right \}$.
\item Compute $l_{\mathcal{G}}({\bf x})$ and $u_{\mathcal{G}}({\bf x})$ defined in \eqref{eq:lifetimes1} for all $\mathcal{G} \in \bigcup \limits_{t=1}^{n-K} \mathcal{C}^{(t)}({\bf x})$ as follows: 
\begin{enumerate} 
\item  Let $l_{\{i\}}({\bf x}) = 1$ for all $i$. Let $l_{\mathcal{G}_{new}^{(t)}({\bf x})}({\bf x}) = t$ for all $2 \leq t < n-K$. 
\item For all $\mathcal{G} \in \bigcup \limits_{t=1}^{n-K} \mathcal{C}^{(t)}({\bf x})$, initialize $u_{\mathcal{G}}({\bf x}) = n-K$. 
\item For $t' = 1, 2, \ldots, n-K$, update $u_{\mathcal{W}_1^{(t')}({\bf x})}({\bf x}) = t'$ and  $u_{\mathcal{W}_2^{(t')}({\bf x})}({\bf x}) = t'$.
\end{enumerate} 
\item For all $\{\{i\}, \{i'\}\} \in \mathcal{L}_1({\bf x})$, compute $\mathcal{\hat S}_{\{i\}, \{i'\}}$ defined in \eqref{eq:Ssub} as follows:
\begin{enumerate} 
\item Use $l_{\{i\}}({\bf x})$ and $l_{\{i'\}}({\bf x})$ to compute $l_{\{i\}, \{i'\}}({\bf x})$ and likewise for $u_{\{i\}, \{i'\}}({\bf x})$, according to \eqref{eq:lifetimes2}. 
\item Compute $\mathcal{M}_{\{i\}, \{i'\}}({\bf x})$ in \eqref{eq:m} by checking which elements of $\mathcal{M}({\bf x})$ are in $[l_{\{i\}, \{i'\}}({\bf x}), u_{\{i\}, \{i'\}}({\bf x})]$. 
\item Use $\mathcal{M}_{\{i\}, \{i'\}}({\bf x})$ and $u_{\{i\}, \{i'\}}({\bf x})$ to compute $h_{\{i\}, \{i'\}}({\bf x})$ according to \eqref{eq:h}. 
\item Compute the coefficients corresponding to the quadratic function $d\left (\{i\}, \{i'\}; {\bf x}'(\phi) \right )$ in constant time, using Lemma \ref{prop:quad}. 
\item Use these coefficients and $h_{\{i\}, \{i'\}}({\bf x})$ to evaluate $\mathcal{\hat S}_{\{i\}, \{i'\}} = \Big \{\phi \geq 0: d (\{i\}, \{i'\}; {\bf x}'(\phi)) > h_{\{i\}, \{i'\}}({\bf x})  ) \Big \}$. This amounts to solving a quadratic inequality in $\phi$.
\end{enumerate} 
\item For all $t = 2, \ldots, n-K$ and for all $\{\mathcal{G}_{new}^{(t)}({\bf x}), \mathcal{G}'\}\in \mathcal{L}_t({\bf x})$, compute  $\mathcal{\hat S}_{\mathcal{G}_{new}^{(t)}({\bf x}), \mathcal{G}'}$ defined in \eqref{eq:Ssub} as follows: 
\begin{enumerate} 
\item Use $l_{\mathcal{G}_{new}^{(t)}({\bf x})}({\bf x})$ and $l_{\mathcal{G}'}({\bf x})$ to compute $l_{\mathcal{G}_{new}^{(t)}({\bf x}), \mathcal{G}'}({\bf x})$ and likewise for $u_{\mathcal{G}_{new}^{(t)}({\bf x}), \mathcal{G}'}({\bf x})$, according to \eqref{eq:lifetimes2}. 
\item Compute $\mathcal{M}_{\mathcal{G}_{new}^{(t)}({\bf x}), \mathcal{G}'}({\bf x})$  in \eqref{eq:m} by checking which elements of $\mathcal{M}({\bf x})$ are between $l_{\mathcal{G}_{new}^{(t)}({\bf x}), \mathcal{G}'}({\bf x})$ and $u_{\mathcal{G}_{new}^{(t)}({\bf x}), \mathcal{G}'}({\bf x})$. 
\item Use $\mathcal{M}_{\mathcal{G}_{new}^{(t)}({\bf x}), \mathcal{G}'}({\bf x})$ and $u_{\mathcal{G}_{new}^{(t)}({\bf x}), \mathcal{G}'}({\bf x})$ to compute $h_{\mathcal{G}_{new}^{(t)}({\bf x}), \mathcal{G}'}({\bf x})$ in \eqref{eq:h}. 
\item Compute the coefficients corresponding to the quadratic function $d\left (\mathcal{G}_{new}^{(t)}({\bf x}), \mathcal{G}'; {\bf x}'(\phi)\right )$ in constant time, using \eqref{eq:lance1} (Proposition \ref{prop:average}). 
\item Use these coefficients and $h_{\mathcal{G}_{new}^{(t)}({\bf x}), \mathcal{G}'}({\bf x})$ to evaluate \\$\mathcal{\hat S}_{\mathcal{G}_{new}^{(t)}({\bf x}), \mathcal{G}'} = \Big \{\phi \geq 0: d \Big (\mathcal{G}_{new}^{(t)}({\bf x}), \mathcal{G}'; {\bf x}'(\phi) \Big ) > h_{\mathcal{G}_{new}^{(t)}({\bf x}), \mathcal{G}'}({\bf x}) \Big \}$. This amounts to solving a quadratic inequality in $\phi$.
\end{enumerate} 
\item Compute $\mathcal{\hat S} = \bigcap \limits_{t=1}^{n-K} \bigcap \limits_{\{ \mathcal{G}, \mathcal{G}' \} \in \mathcal{L}_t ({\bf x})} \mathcal{\hat S}_{\mathcal{G}, \mathcal{G}'}$.
\end{enumerate} 
To see that this algorithm computes $\mathcal{\hat S}$ in $\mathcal{O}\Big ((|\mathcal{M}({\bf x})| + \log(n))n^2 \Big )$ time,
observe that: 
\begin{itemize} 
\item Steps 1 and 2 can each be performed in $\mathcal{O}(n)$ time. Note that $\left | \bigcap \limits_{t=1}^{n-K} \mathcal{C}^{(t)} \right | = \mathcal{O}(n)$. 
\item Each of the $\mathcal{O}(n^2)$ iterations in Step 3 can be performed in $\mathcal{O}(|\mathcal{M}({\bf x})| + 1)$ time. 
\item Each of the $\mathcal{O}(n^2)$  iterations in Step 4 can be performed in $\mathcal{O}(|\mathcal{M}({\bf x})| + 1)$ time.
\item Step 5 can be performed in $\mathcal{O}(n^2 \log(n))$ time. This is because $\mathcal{\hat S}_{\mathcal{G}, \mathcal{G}'}$ solves a quadratic inequality for each $\{\mathcal{G}, \mathcal{G}'\} \in \mathcal{L}({\bf x})$, we can take the intersection over the solution sets of $N$ quadratic inequalities in $\mathcal{O}(N \log N)$ time \citep{bourgon2009intervals}, and $\left | \mathcal{L}({\bf x}) \right | = \mathcal{O}(n^2)$. 
\end{itemize} 

\section{Supplementary material for Section \ref{sec:unknown}}
\label{sec:supp-unknown}
The goal of this section is to establish an estimator of $\sigma$ that satisfies the condition from Theorem \ref{thm:plug-in-conservative}. For any data set ${\bf x}$, define 
\begin{align} 
\hat \sigma({\bf x}) = \sqrt{\frac{1}{nq - q} \sum \limits_{i=1}^n  \sum \limits_{j=1}^q ( x_{ij} - \bar{x}_j)^2},\label{eq:def-sigma} 
\end{align} 
where $\bar{x}_j = \sum \limits_{i=1}^n x_{ij}/n$ for all $j = 1, 2, \ldots, q$. 
We will first show that under some assumptions, $\hat \sigma({\bf X}^{(n)})$ asymptotically over-estimates $\sigma$. Then, we will show that if we estimate $\sigma$ by applying $\hat \sigma(\cdot)$ to an independent and identically distributed copy of ${\bf X}^{(n)}$, then the condition from Theorem \ref{thm:plug-in-conservative} is satisfied. 

Let $K^* > 1$ be unknown. We assume that the sequence of mean matrices for $\{{\bf X}^{(n)}\}_{n=1}^\infty$ in Theorem \ref{thm:plug-in-conservative}, which we denote as $\left \{\bm \mu^{(n)} \right \}^{\infty}_{n=1}$, satisfies the following assumptions.
\begin{assumption} 
For all $n = 1, 2, \ldots$, $\left \{\mu^{(n)}_i \right \}_{i=1}^n = \{\theta_1, \ldots, \theta_{K^*}\}$, i.e. there are exactly $K^*$ distinct mean vectors among the first $n$ observations. 
\label{Kmeans} 
\end{assumption} 
\begin{assumption} \label{mixing} 
For all $k = 1, 2, \ldots, K^*$, 
$ \underset{n \rightarrow \infty}{\lim} \left ( \sum \limits_{i=1}^n \mathds{1} \{\mu^{(n)}_i = \theta_k\}/n  \right ) = \lambda_k,$
where $\sum \limits_{k=1}^{K^*} \lambda_k = 1$ and $\lambda_k > 0$. That is, the proportion of the  first $n$ observations that have mean vector $\theta_k$ converges to $\lambda_k$ for all $k$.  
\end{assumption} 
This leads to the following result, which says that under Assumptions \ref{Kmeans} and \ref{mixing}, $\hat \sigma({\bf X}^{(n)})$ asymptotically over-estimates $\sigma$. 
\begin{lemma} \label{lem:marginal}
For $n = 1, 2, \ldots$, suppose that ${\bf X}^{(n)} \sim \mathcal{MN}_{n \times q}(\bm \mu^{(n)}, {\bf I}_n, \sigma^2 {\bf I}_q)$, and that Assumptions \ref{Kmeans} and \ref{mixing} hold for some $K^* > 1$. Then,  
\begin{align} 
\underset{n\rightarrow\infty}{\lim} ~ \mathbb{P}(\hat \sigma({\bf X}^{(n)}) \geq \sigma) = 1. \label{eq:marginal-overshoot} 
\end{align} 
where $\hat \sigma(\cdot)$ is defined in \eqref{eq:def-sigma}. 
\end{lemma} 

\begin{proof} 
Observe that 
\begin{align} \left ( \frac{nq-q}{nq} \right ) \hat \sigma^2({\bf X}^{(n)}) =\frac{1}{nq} \sum \limits_{i=1}^n  \sum \limits_{j=1}^q \left ( X_{ij}^{(n)} \right ) ^2 - \frac{1}{q} \sum \limits_{j=1}^q \left ( \bar{X}^{(n)}_j \right ) ^2. \label{eq:factorize} 
\end{align} 
We will first show that 
\begin{align} 
\frac{1}{nq} \sum \limits_{i=1}^n  \sum \limits_{j=1}^q \left ( X^{(n)}_{ij} \right ) ^2 \overset{p}{\longrightarrow} \sigma^2 + \frac{1}{q} \sum \limits_{j=1}^q  \sum \limits_{k=1}^{K^*}  \lambda_k \theta_{kj}^2. \label{eq:converge1} 
\end{align} 
Since for any positive integer $r$ and for any $j = 1, 2, \ldots, q$, we have
\begin{align} 
\underset{n \rightarrow \infty}{\lim} ~ \frac{1}{n} \sum \limits_{i=1}^n \left [ \mu_{ij}^{(n)} \right ]^r &=  \underset{n \rightarrow \infty}{\lim} ~ \frac{1}{n} \sum \limits_{i=1}^n  \sum \limits_{k=1}^{K^*} \mathds{1} \left \{\mu_{i}^{(n)} = \theta_k \right \} \left ( \theta_{kj} \right ) ^r = \sum  \limits_{k=1}^{K^*}  \lambda_k \left ( \theta_{kj} \right )^r, \label{eq:mixasym} 
\end{align} 
it follows that 
\begin{align} 
\underset{n \rightarrow \infty}{\lim} ~ \mathbb{E}  \left [ \frac{1}{nq} \sum \limits_{i=1}^n  \sum \limits_{j=1}^q \left ( X^{(n)}_{ij} \right ) ^2 \right ] &= \underset{n \rightarrow \infty}{\lim}  ~ \frac{1}{nq} \sum \limits_{i=1}^n  \sum \limits_{j=1}^q\left  [\text{Var}\left [ X^{(n)}_{ij} \right ]  + \left ( \mathbb{E}\left [ X^{(n)}_{ij}  \right ] \right )^2 \right ] \nonumber \\ 
&= \underset{n \rightarrow \infty}{\lim}  \left ( \sigma^2 + \frac{1}{nq} \sum \limits_{i=1}^n \sum \limits_{j=1}^q \left [ \mu_{ij}^{(n)} \right ]^2 \right ) \nonumber \\ 
&= \sigma^2 + \frac{1}{q} \sum \limits_{j=1}^q \sum  \limits_{k=1}^{K^*}  \lambda_k \left ( \theta_{kj} \right )^2. \label{eq:mean}
\end{align} 
Furthermore, 
\begin{align} 
 \text{Var} \left [ \frac{1}{nq} \sum \limits_{i=1}^n  \sum \limits_{j=1}^q \left ( X^{(n)}_{ij} \right ) ^2 \right ] &\leq \frac{1}{(nq)^2} \sum \limits_{i=1}^n \sum \limits_{j=1}^q \mathbb{E}\left [\left ( X^{(n)}_{ij} \right )^4 \right ]  \nonumber \\ 
&=\frac{1}{(nq)^2} \sum \limits_{i=1}^n \sum \limits_{j=1}^q \left [\left  (\mu_{ij}^{(n)} \right )^4  + 6 \left  (\mu_{ij}^{(n)} \right )^2 \sigma^2 + 3\sigma^4 \right ]   \label{eq:normal-fact}\\
&= \frac{1}{nq} \left [ \frac{1}{nq} \sum \limits_{i=1}^n \sum \limits_{j=1}^q \left ( \mu_{ij}^{(n)} \right )^4 \right ] + \frac{6 \sigma^2}{nq} \left [ \frac{1}{nq} \sum \limits_{i=1}^n \sum \limits_{j=1}^q \left [ \mu_{ij}^{(n)} \right ]^2 \right ]  + \frac{3\sigma^4}{nq}, \nonumber
\end{align}
where \eqref{eq:normal-fact} follows from the fact that $X_{ij}^{(n)} \sim N(\mu_{ij}^{(n)}, \sigma^2)$. It follows from \eqref{eq:mixasym} that \begin{align} 
\underset{n\rightarrow \infty}{\lim} ~ \text{Var} \left [ \frac{1}{nq} \sum \limits_{i=1}^n  \sum \limits_{j=1}^q \left ( X^{(n)}_{ij} \right ) ^2 \right ] = 0. \label{eq:var}
\end{align}
Equation \eqref{eq:converge1} follows from \eqref{eq:mean} and \eqref{eq:var}. 
Next, since \eqref{eq:mixasym} holds for $r = 1$ and for all $j = 1, 2, \ldots, q$, we have that 
\begin{gather*} 
\underset{n \rightarrow \infty}{\lim} ~\mathbb{E}\left [  \bar{X}^{(n)}_j \right ] = \underset{n \rightarrow \infty}{\lim} ~\frac{1}{n} \sum \limits_{i=1}^n \mu_{ij}^{(n)} = \sum  \limits_{k=1}^{K^*}  \lambda_k \theta_{kj}, \quad \quad \underset{n \rightarrow \infty}{\lim}  \text{Var}\left [ \bar{X}^{(n)}_j \right ] = \underset{n \rightarrow \infty}{\lim} \frac{\sigma^2}{n} = 0. 
\end{gather*} 
Thus, for all $j = 1, 2, \ldots, q$, we have that $\bar{X}^{(n)}_j \overset{p}{\rightarrow} \sum \limits_{k=1}^{K^*} \lambda_k \theta_{kj}$. Combining this fact with \eqref{eq:converge1} and \eqref{eq:factorize} yields: 
\begin{align} 
\left ( \frac{nq-q}{nq} \right ) \hat \sigma^2({\bf X}^{(n)}) \overset{p}{\rightarrow} \sigma^2 + \frac{1}{q} \sum \limits_{j=1}^q \sum \limits_{k=1}^{K^*}  \lambda_k (\theta_{kj} - \tilde \theta_j)^2, \label{eq:square-overshoot}
\end{align} 
where $\tilde \theta_j = \sum \limits_{k=1}^{K^*} \lambda_k \theta_{kj}$ for all $j = 1, 2, \ldots, q$. It follows from
 \eqref{eq:square-overshoot} that 
\begin{align} 
\left ( \frac{nq-q}{nq} \right ) \hat \sigma^2({\bf X}^{(n)}) \overset{p}{\rightarrow} \sigma^2 + c, \label{eq:converge2} 
\end{align}
for $c = \frac{1}{q} \sum \limits_{j=1}^q \sum \limits_{k=1}^{K^*}  \lambda_k (\theta_{kj} - \tilde \theta_j)^2 > 0$. Applying the continuous mapping theorem to \eqref{eq:converge2} yields 
$ \hat \sigma({\bf X}^{(n)}) \overset{p}{\rightarrow} \sqrt{\sigma^2 + c}, $
where $\sqrt{\sigma^2 + c} > \sigma$.  Therefore, 
$ \underset{n \rightarrow \infty}{\lim} \mathbb{P} \left (\hat \sigma({\bf X}^{(n)}) \geq \sigma \right ) = 1.$
\end{proof} 
The following result establishes an estimator of $\sigma$ that satisfies the condition in Theorem \ref{thm:plug-in-conservative}, by making use of an independent and identically distributed copy of ${\bf X}^{(n)}$, if such a copy is available. If not, then sample-splitting can be used.
\begin{corollary} 
\label{cor:sigma} 
For $n = 1, 2, \ldots$, suppose that ${\bf X}^{(n)} \sim \mathcal{MN}_{n \times q}(\bm \mu^{(n)}, {\bf I}_n, \sigma^2 {\bf I}_q)$, and let ${\bf x}^{(n)}$ be a realization from ${\bf X}^{(n)}$, and that Assumptions 1 and 2 hold for some $K^* > 1$. For all $n = 1, 2, \ldots, $ let  $\hc_1^{(n)}$ and $\hc_2^{(n)}$ be a pair of estimated clusters in  $\mathcal{C}({\bf x}^{(n)})$. For all $n = 1, 2, \ldots, $ let ${\bf Y}^{(n)}$ be an independent and identically distributed copy of ${\bf X}^{(n)}$. Then,
\begin{align} \Scale[0.9]{\underset{m,n \rightarrow \infty}{\lim} ~\mathbb{P}_{H_0^{ \left \{\hc_1^{(n)}, \hc_2^{(n)} \right \}}}\left (\hat p({\bf X}^{(n)}, \hat \sigma({\bf Y}^{(m)})) \leq \alpha ~\big | ~ \hc_1^{(n)}, \hc_2^{(n)} \in \mathcal{C}({\bf X}^{(n)}) \right ) \leq \alpha,  \text{ for all } 0 \leq \alpha \leq 1}.  \label{eq:conditional2} 
\end{align} 
\end{corollary} 
\begin{proof} 
By the independence of ${\bf X}^{(n)}$ and ${\bf Y}^{(n)}$ for all $n = 1, 2, \ldots$, and by Lemma \ref{lem:marginal}, 
$$\underset{m,n\rightarrow\infty}{\lim} ~ \mathbb{P}_{H_0^{ \left \{\hc_1^{(n)}, \hc_2^{(n)} \right \}}}\left ( \hat \sigma({\bf Y}^{(m)}) > \sigma ~ \big | ~ \hc_1^{(n)}, \hc_2^{(n)} \in \mathcal{C}({\bf X}^{(n)}) \right ) = 1.$$ Thus, \eqref{eq:conditional2} follows from Theorem \ref{thm:plug-in-conservative}.
\end{proof} 

\section{Additional simulation studies} 

\subsection{Null p-values when $\sigma$ is unknown}
\label{sec:supp-sim}

In Section \ref{sec:unknown}, we propose plugging an estimator of $\sigma$ into the p-value defined in \eqref{eq:pval-def} to get \eqref{eq:est-pval}, and we established conditions under which rejecting $H_0^{\{\hc_1, \hc_2\}}$ when \eqref{eq:est-pval} is below $\alpha$ asymptotically controls the selective type I error rate (Theorem \ref{thm:plug-in-conservative}). 
Furthermore, in Section \ref{sec:unknown}, we establish an estimator of $\sigma$ that satisfisfies the conditions in Theorem \ref{thm:plug-in-conservative} (Corollary \ref{cor:sigma}). In the following, we will investigate the finite-sample implications of applying the approach outlined in Section \ref{sec:supp-unknown}.  

We generate data from \eqref{eq:mod} with $n = 200$, $q = 10$, $\sigma = 1$, and two clusters,
\begin{align} 
\mu_1 = \cdots = \mu_{100} = \left [ \begin{smallmatrix} \delta/2 \\ 0_{q-1} \end{smallmatrix} \right ] ,  \mu_{101} = \cdots = \mu_{200} = \left [ \begin{smallmatrix} 0_{q-1} \\ -\delta/2 \end{smallmatrix} \right ], \label{eq:equi-means-supp}
\end{align} 
for $\delta \in \{2, 4, 6\}$. We randomly split each data set into equally sized ``training" and ``test" sets. For each training set, we use average, centroid, single, and complete linkage hierarchical clustering to estimate three clusters, and then test $H_0^{\{\hc_1, \hc_2\}}: \bar{\mu}_{\hc_1} = \bar{\mu}_{\hc_2}$ for a randomly chosen pair of clusters. We estimate $\sigma$ by applying $\hat \sigma (\cdot)$ in \eqref{eq:def-sigma} to the corresponding test set, then use it to compute p-values for  average, centroid, and single linkage. In the case of complete linkage, we approximate \eqref{eq:est-pval} by replacing $\sigma$ in the importance sampling procedure described in Section \ref{sec:approx} with $\hat \sigma (\cdot)$ in \eqref{eq:def-sigma} applied to the corresponding test set. Figure \ref{fig:type1-supp} displays QQ plots of the empirical distribution of the p-values against the Uniform$(0, 1)$ distribution, over 500 simulated data sets where $H_0^{\{\hc_1, \hc_2\}}: \bar{\mu}_{\hc_1} = \bar{\mu}_{\hc_2}$ holds, for $\delta \in \{2, 4, 6\}$.  

\begin{figure}[h!]
\centering
\includegraphics[scale=0.5]{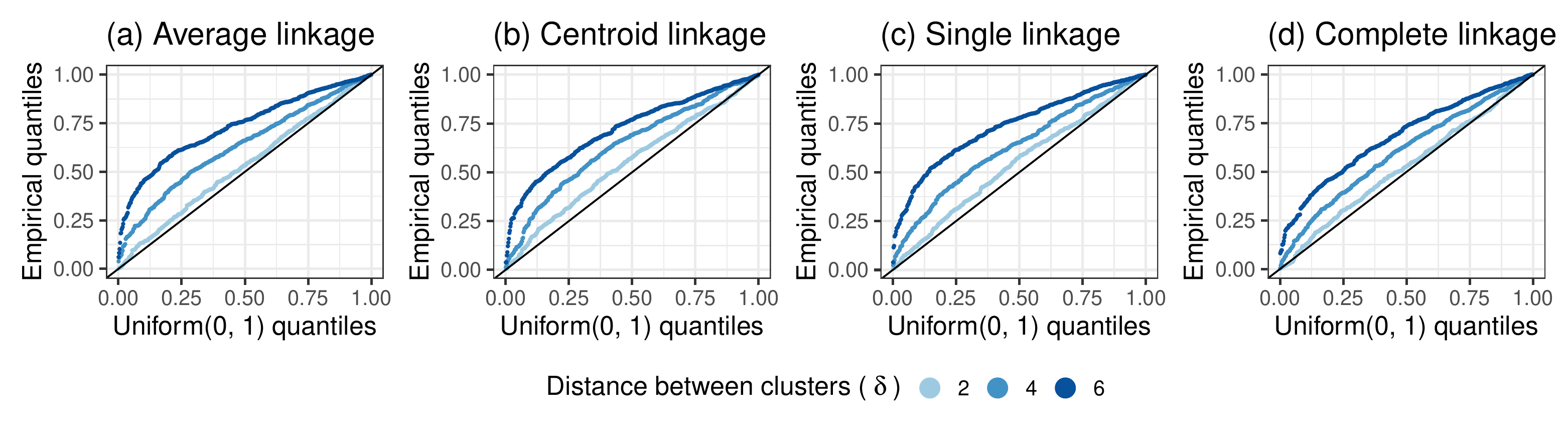}
\caption{\label{fig:type1-supp} For the simulation study described in Section \ref{sec:supp-sim}, QQ-plots of the p-values, using (a) average linkage, (b) centroid linkage, (c) single linkage, and (d) complete linkage, over 500 simulated data sets where $H_0^{\{\hc_1, \hc_2\}}: \bar{\mu}_{\hc_1} = \bar{\mu}_{\hc_2}$ holds, for $\delta \in \{2, 4, 6\}$. } 
\end{figure} 

The p-values appear to be stochastically larger than the Uniform$(0, 1)$ distribution, across all linkages and all values of $\delta$ in \eqref{eq:equi-means-supp}. As $\delta$ decreases and the clusters become less separated,  $\hat \sigma (\cdot)$ in \eqref{eq:def-sigma} overestimates $\sigma$ less severely. Thus, as $\delta$ decreases, the distribution of the p-values becomes closer to the Uniform$(0, 1)$ distribution.

\subsection{Power as a function of effect size}
\label{sec:effect}
Recall that in Section \ref{sec:power}, we considered the \emph{conditional} power of the test proposed in Section \ref{sec:frame-test} to detect a difference in means between two true clusters. We now evaluate the power of the test proposed in Section \ref{sec:frame-test} to detect a difference in means between estimated clusters, without conditioning on having correctly estimated the true clusters. 

We generate data from \eqref{eq:mod} with $n = 150$, $\sigma = 1$, $q = 10$, and $\mu$ given by \eqref{eq:equi-means}. We simulate $10,000$ data sets for nine evenly-spaced values of $\delta \in [3, 7]$. For each simulated data set, we test $H_0^{\{\hc_1, \hc_2\}}: \bar{\mu}_{\hc_1} = \bar{\mu}_{\hc_2}$ for a randomly chosen pair of clusters obtained from single, average, centroid, and complete linkage hierarchical clustering, with significance level $\alpha = 0.05$. We define the \emph{effect size} to be 
$\Delta = \|\bar{\mu}_{\hc_1} - \bar{\mu}_{\hc_2} \|_2/\sigma$
and consider the probability of rejecting $H_0^{\{\hc_1, \hc_2\}}:  \bar{\mu}_{\hc_1} = \bar{\mu}_{\hc_2}$ as a function of $\Delta$.  We smooth our power estimates by fitting a regression spline using the \verb+gam+ function in the \verb+R+ package \verb+mgcv+ \citep{wood2015package}. Figure \ref{fig:power}(a)--(d) displays the smoothed estimates when $\min\{|\hc_1|, |\hc_2|\} \geq 10$, and Figure \ref{fig:power}(e)--(h) displays the smoothed estimates when $\min\{|\hc_1|, |\hc_2|\} < 10$. (We stratify our results because the power is much lower when $\min\{|\hc_1|, |\hc_2|\}$ is small.)

\begin{figure}[h!] 
\centering
\includegraphics[scale=0.55]{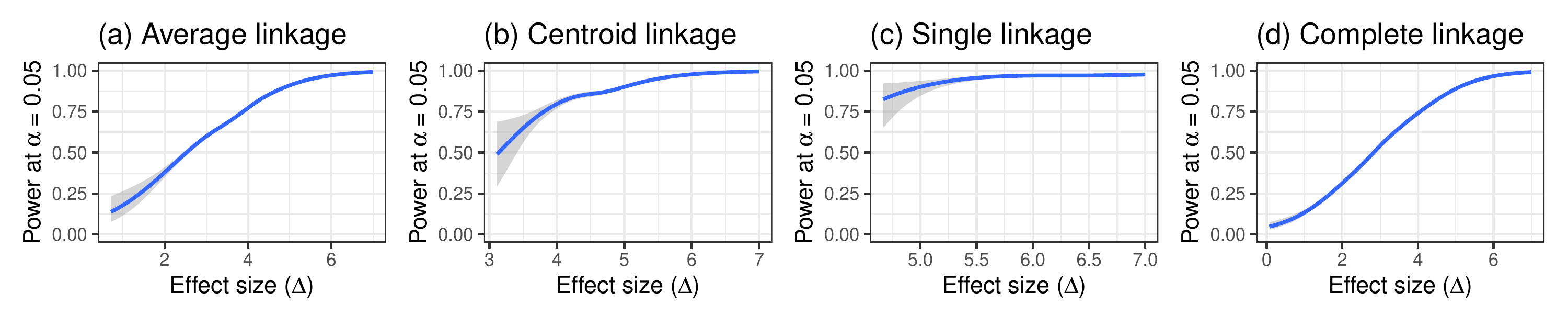} \\ 
\includegraphics[scale=0.55]{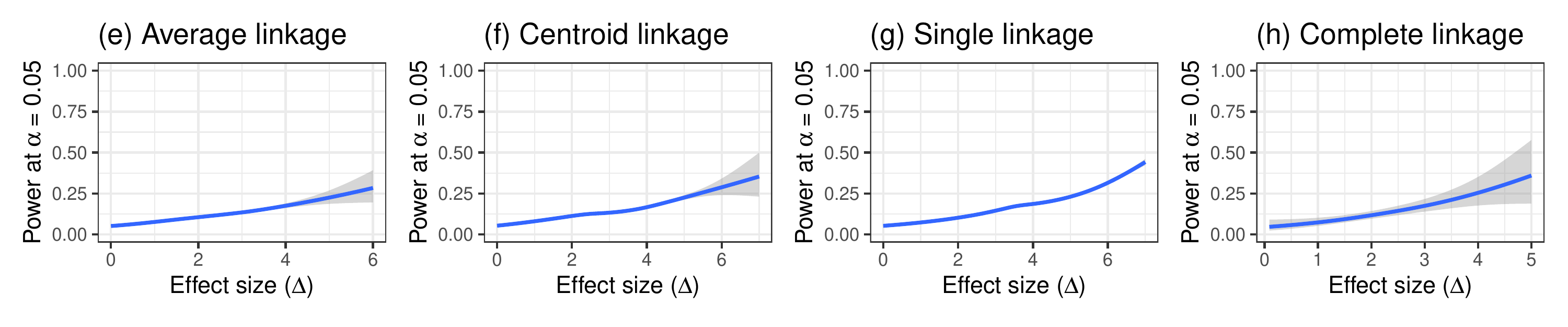} 
\caption{\label{fig:power}  For the simulation study described in Section \ref{sec:effect}, smoothed power estimates as a function of the effect size $\Delta = \|\bar{\mu}_{\hc_1} - \bar{\mu}_{\hc_2} \|_2/\sigma$. In (a)--(d),  $\min\{|\hc_1|, |\hc_2|\} \geq 10$, and in (e)--(h), $\min\{|\hc_1|, |\hc_2|\} < 10$. } 
\end{figure} 

It may come as a surprise that the $x$-axis on Figure \ref{fig:power}(c) starts at 4.5, while the $x$-axis on Figure \ref{fig:power}(d) starts at 0. This is because single linkage hierarchical clustering produces unbalanced clusters unless it successfully detects the true clusters. Thus, the condition $\min \{|\hc_1|, |\hc_2|\} \geq 10$ is only satisfied for single linkage hierarchical clustering when the true clusters are detected, which happens when $\Delta$ is greater than $4.5$. The $x$-axis on Figure (b) starts at 3 rather than 0 for a similar reason. 

In Figure \ref{fig:power}, for all four linkages, the power to reject $H_0^{\{\hc_1, \hc_2\}}: \bar{\mu}_{\hc_1} = \bar{\mu}_{\hc_2}$  increases as the effect size ($\Delta$) increases. 
All four linkages have similar power where the $x$-axes overlap. 

\begin{figure}[h!] 
\centering
\includegraphics[scale=0.4]{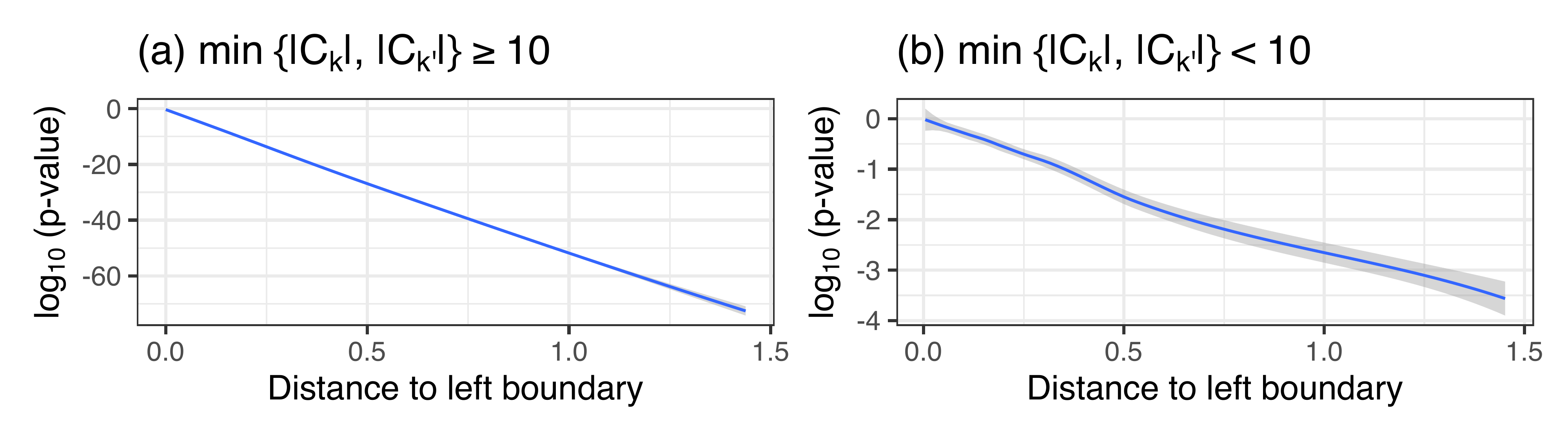}
\caption{\label{fig:boundary}  For the simulation study described in Section \ref{sec:effect} with average linkage and $\Delta = \|\bar{\mu}_{\hc_1} - \bar{\mu}_{\hc_2} \|_2/\sigma = 5$, smoothed p-values on the $\log_{10}$ scale as a function of the distance between the test statistic and the left boundary of $\mathcal{\hat S}$, defined in \eqref{eq:defS}. }
\end{figure} 

Our test has low power when the test statistic ($\|\bar{x}_{\hc_1} - \bar{x}_{\hc_2} \|_2$) is very close to the left boundary of $\mathcal{\hat S}$ defined in \eqref{eq:defS}. This is a direct result of the fact that $p({\bf x}; \{\hc_1, \hc_2\}) = \mathbb{P}(\phi \geq \|\bar{x}_{\hc_1} - \bar{x}_{\hc_2} \|_2 \mid \phi \in \mathcal{\hat S})$, for $\phi \sim (\sigma \sqrt{\frac{1}{|\hc_1} + \frac{1}{|\hc_2|}})  \cdot \chi_q$ (Theorem \ref{thm:pval}). Figure \ref{fig:boundary} displays the smoothed p-values from average linkage hierarchical clustering with effect size $\Delta = \|\bar{\mu}_{\hc_1} - \bar{\mu}_{\hc_2} \|_2/\sigma = 5$ as a function of the distance between the test statistic and the left boundary of $\mathcal{\hat S}$. In Figure \ref{fig:boundary}, the p-values increase as the distance increases. Similar results have been observed in other selective inference frameworks; see \citet{kivaranovic2020length} for a detailed discussion.

%\appendix 

%\input{paper-selective-appendix}

\end{document}